\documentclass[10pt,a4paper]{article}
\usepackage[latin1]{inputenc}
\usepackage{setspace}
\usepackage{amsmath}
\usepackage{amsfonts}
\usepackage{amssymb}
\usepackage{amsthm}
\usepackage{bm}
\usepackage{graphicx}
\usepackage{enumitem}
\graphicspath{ {./figures/} }

\usepackage[left=1in, right=1in]{geometry}
\usepackage{authblk}
\usepackage{xcolor}

\usepackage{hyperref}

\usepackage[mathscr]{euscript}
\usepackage{bbm}

\usepackage{enumitem}

\theoremstyle{plain}
\newtheorem{theorem}{Theorem}
\newtheorem{lemma}{Lemma}
\newtheorem{corollary}{Corollary}
\newtheorem{proposition}{Proposition}

\theoremstyle{remark}
\newtheorem{example}{Example}
\newtheorem{condition}{Condition}

\usepackage{caption}
\usepackage{algorithm}
\usepackage{algorithmicx}
\usepackage[noend]{algpseudocode}

\DeclareMathOperator*{\argmax}{argmax}
\DeclareMathOperator*{\argmin}{argmin}

\newcommand{\real}{{\mathbb{R}}}
\newcommand{\modelspace}{{\mathcal{M}}}
\newcommand{\observation}{\mathbf{X}}
\newcommand{\observationspace}{{\mathcal{X}}}
\newcommand{\coarsen}{\mathcal{C}}

\newcommand{\decisionspace}{{\mathcal{D}}}

\newcommand{\expect}{{\mathbb{E}}}
\newcommand{\Var}{{\mathrm{Var}}}
\newcommand{\intd}{{\mathrm{d}}}
\newcommand{\MV}{{\mathrm{MV}}}
\newcommand{\ind}{{\mathbbm{1}}}

\usepackage{tikz}
\usetikzlibrary{positioning}

\tikzset{%
    every neuron/.style={
        circle,
        draw,
        minimum size=1cm
    },
    neuron missing/.style={
        draw=none, 
        scale=4,
        text height=0.333cm,
        execute at begin node=\color{black}$\vdots$
    },
}

\usepackage{natbib}
\title{Adversarial Meta-Learning of Gamma-Minimax Estimators\\That Leverage Prior Knowledge}
\author[1]{Hongxiang Qiu}
\author[2]{Alex Luedtke}
\affil[1]{Department of Epidemilogy and Biostatistics, Michigan State University}
\affil[2]{Department of Statistics, University of Washington}

\date{}

\begin{document}
\maketitle

\begin{abstract}
    Bayes estimators are well known to provide a means to incorporate prior knowledge that can be expressed in terms of a single prior distribution. However, when this knowledge is too vague to express with a single prior, an alternative approach is needed. Gamma-minimax estimators provide such an approach. These estimators minimize the worst-case Bayes risk over a set $\Gamma$ of prior distributions that are compatible with the available knowledge. Traditionally, Gamma-minimaxity is defined for parametric models. In this work, we define Gamma-minimax estimators for general models and propose adversarial meta-learning algorithms to compute them when the set of prior distributions is constrained by generalized moments. Accompanying convergence guarantees are also provided. We also introduce a neural network class that provides a rich, but finite-dimensional, class of estimators 
    from which a Gamma-minimax estimator can be selected. 
    We illustrate our method in two settings, namely entropy estimation and a prediction problem that arises in biodiversity studies.
\end{abstract}

\section{Introduction}

A variety of principles can be used to guide the search for a suitable statistical estimator. Asymptotic efficiency \citep{Pfanzagl1990}, minimaxity \citep{Wald1945} and Bayes optimality \citep{Berger1985} are popular examples of such principles. Defining the performance criteria underlying these principles requires specifying a model space, that is, a collection of possible data-generating mechanisms known to contain the true, underlying distribution.

It is often desirable to incorporate prior information about the true data-generating mechanism into a statistical procedure. This might be done differently in different statistical paradigms. %
For frequentist methods, such as those based on the asymptotic efficiency or minimax principle, the primary way to incorporate this information is via the definition of the model space. 
In the Bayesian paradigm, such information may be represented by further specifying a prior distribution (or \emph{prior} for short) over the model space and aiming for an estimator that minimizes the induced Bayes risk. However, in many cases, there may be several priors that are compatible with the available information, especially in the context of rich model spaces. %

The Gamma-minimax paradigm, proposed by \cite{Robbins1951}, provides a principled means to overcome the challenge of specifying a single prior distribution. Under this paradigm, a statistician first specifies a set $\Gamma$ of all priors that are consistent with the available prior information and subsequently seeks an estimator that minimizes the worst-case Bayes risk over this set of priors. The Gamma-minimax paradigm may be viewed as a robust version of the Bayesian paradigm that is less sensitive to misspecification of a prior distribution \citep{Vidakovic2000}. When it is infeasible to specify a prior due to the complexity of the model space, the Gamma-minimax paradigm may also be viewed as a feasible substitute for the Bayesian paradigm. The Gamma-minimax paradigm is closely related to Bayes and minimax paradigms: when the set of priors consists of a single prior, a Gamma-minimax estimator is Bayes with respect to that prior; when the set $\Gamma$ of priors is the entire set of possible prior distributions, a Gamma-minimax estimator is also minimax.

Gamma-minimax estimators have been studied for a variety of problems. Some explicit forms of Gamma-minimax estimators have been obtained. For example, \cite{olman1985} studied Gamma-minimax estimation of the mean of a normal distribution for the set of symmetric and unimodal priors on an interval and obtained an explicit form when this interval is sufficiently small. \cite{Eichenauer-Herrmann1990} generalized this result to more general parametric models and \cite{Eichenauer-Herrmann1994} obtained a further generalization with the requirement of symmetry on the priors dropped. \cite{Chen1988} studied Gamma-minimax estimation for multinomial distributions and the set of priors with bounded mean. \cite{chen1991} studied Gamma-minimax estimation for one-parameter exponential families and the set of priors that place certain bounds on the first two moments. These results do not deal with general model spaces, such as semiparametric or nonparametric models, and general forms of the set of priors that may not directly impose bounds on prior moments of the parameters of interest. One reason for this lack of generality might be that, in the existing literature, Gamma-minimaxity is defined only for parametric models. However, an issue with parametric models is that they often fail to contain the true data-generating mechanism, in which case output from the aforementioned statistical procedures may no longer be interpretable. Another possible reason is that it is typically intractable to analytically derive Gamma-minimax estimators, even for parametric models.

To overcome this lack of analytical tractability, meta-learning algorithms to compute a minimax or Gamma-minimax estimator have been proposed. Still, most of these works focus on parametric models. For example, \cite{Nelson1966} and \cite{Kempthorne1987} each proposed an algorithm to compute a minimax estimator. \cite{Bryan2007} and \cite{Schafer2009} proposed an algorithm to compute an approximate confidence region of optimal expected size in the minimax sense. \cite{Noubiap2001} proposed an iterative algorithm to compute a Gamma-minimax decision for the set of priors constrained by generalized moment conditions. More recent works explored computing estimators under more general models. For example, \cite{Luedtke2020} introduced an approach, termed Adversarial Monte Carlo meta-learning (AMC), for constructing minimax estimators. In the special case of prediction problems with mean-squared error, \cite{Luedtke2020adversarial_prediction} studied the invariance properties of the decision problem and their implications for AMC.

In this paper, we make the following contributions:
\begin{enumerate}[resume]
	\item\label{it:iterativeScheme} %
	We propose iterative adversarial meta-learning algorithms for %
	constructing Gamma-minimax estimators for a general model space and class of estimators. We further provide convergence guarantees for these algorithms.
\end{enumerate}
To our best knowledge, this is the first algorithm to compute Gamma-minimax estimators under general models, including infinite-dimensional models. We also show that, for certain problems, there is a unique Gamma-minimax estimator and our computed estimator converges to this estimator as the number of iterations increases to infinity.

Like the approach proposed in \cite{Noubiap2001}, we consider sets of priors characterized by (in)equality constraints on prior generalized moments and our proposed iterative algorithm involves solving a
discretized Gamma-minimax optimization problem in each intermediate step. 
However, we explicitly describe algorithms to solve these minimax problems, which facilitates the use of our approach by practitioners. When the space of estimators can be parameterized by a Euclidean space and gradients are available, we propose to use a gradient-based algorithm or a stochastic variant thereof. When gradients are unavailable, we propose to instead use fictitious play \citep{Brown1951,Robinson1951} to compute a stochastic estimator, which is a mixture of deterministic estimators belonging to some specified collection. We also provide a convergence result that is applicable even when this collection has infinite cardinality. This is in contrast to the results in %
\cite{Robinson1951}, which require that each player has only finitely many possible deterministic strategies.

\begin{enumerate}[resume]
	\item We propose a Markov chain Monte Carlo (MCMC) method to construct the approximating grids defining the discretized Gamma-minimax problems used in our iterative scheme.
\end{enumerate}
Like the approach proposed in \cite{Noubiap2001}, our proposed iterative algorithm relies on increasingly fine finite grids over the model space. However, since we allow the model space to be high or even infinite-dimensional, randomly adding points to the grid may lead to unacceptably slow convergence. To overcome this challenge, we propose to use MCMC to efficiently construct such grids.

Our theoretical results allow for many different choices of classes of estimators. Our final contribution concerns the introduction of one such class:
\begin{enumerate}[resume]
	\item %
	We introduce a new neural network architecture that can be used to parameterize statistical estimators and argue that this class represents an appealing class to optimize over.
\end{enumerate}
For this final point, we build on existing works in adversarial learning \citep[e.g.,][]{Goodfellow2014,Luedtke2020,Luedtke2020adversarial_prediction} and extreme learning machines \citep{Huang2006}. 
Thanks to the universal approximation properties of neural networks \citep[e.g.,][]{Hornik1991,Csaji2001} and extreme learning machines \citep{Huang2006universal}, we also show that both of these parameterizations can achieve good performance for sufficiently large networks. Furthermore, inspired by pre-training \citep[e.g.,][]{Erhan2010} and transfer learning \citep[e.g.,][]{Torrey2009}, we recommend leveraging knowledge of existing estimators as inputs to the network in settings where this is possible. Under such choices of the space of estimators, we can expect to obtain a useful estimator even if the associated nonconvex-concave minimax problems prove to be difficult.

This paper is organized as follows. In Section~\ref{section: setup}, we introduce the framework of Gamma-minimax estimation and regularity conditions that we assume throughout the paper. In Section~\ref{section: algorithm}, we describe our proposed iterative adversarial meta-learning algorithms. In Section~\ref{section: considerations in implementation}, we discuss considerations when choosing hyperparameters in the algorithms. In Section~\ref{section: simulation}, we demonstrate our method in three simulation studies. We conclude with a discussion in Section~\ref{section: discussion}. Proof sketches of key results are provided in the main text, and complete proofs can be found in the appendix. We also provide a table summarizing the frequently used symbols in Table~\ref{table: symbols} in the appendix.
The code for our simulations is available on GitHub \citep{github}.

\section{Problem setup} \label{section: setup}

Let $\modelspace$ be a separable Hausdorff space of data-generating mechanisms that contains the truth $P_0$ and is equipped with a metric $\rho$.
Under a data-generating mechanism $P \in \modelspace$, let $\observation^* \in \observationspace^*$ denote the random data being generated, where $\observationspace^*$ is the space of values that the random data takes. Let $\coarsen$ denote a known coarsening mechanism such that the observed data $\observation=\coarsen(\observation^*) \in \observationspace$, where $\observationspace$ is the space of observed data. In some cases, the coarsening mechanism will be the identity map, whereas in other settings, such as those in which missing, censored or truncated data is present, the coarsening mechanism will be nontrivial \citep[e.g.,][]{Birmingham2003,Gill1997,Heitjan1991,Heitjan1993,Heitjan1994}. %
Let $\decisionspace$ denote the space of estimators (or decision functions) equipped with a metric $\varrho$. In practice, for computational feasibility, we will mainly consider an estimator space $\decisionspace$ that contains estimators parameterized by a Euclidean space such as linear estimators or neural networks, and approximates a more general space $\decisionspace_0$, for example, the space of all estimators satisfying certain smoothness conditions. We discuss considerations concerning the choice of $\decisionspace$ in Section~\ref{section: choice of estimator space} and note that our proposed methods may be applied to broader estimator classes.
We treat $\decisionspace$ as fixed throughout this paper.
Let $R: \decisionspace \times \modelspace \rightarrow \real$ denote a risk function that measures the performance of an estimator under a data-generating mechanism such that smaller risks are preferable. We suppose throughout that $\modelspace$ and $\decisionspace$ are equipped with the topologies induced by $\rho$ and $\varrho$, respectively.

We now present three examples in which we formulate statistical decision problems in the above form. The first example is a general example of point estimation. We use this example to illustrate how the Gamma-minimax estimation framework naturally fits into many statistical problems. The other two examples are more concrete and we will study them in the simulations and data analyses.

\begin{example}[Point estimation] \label{example: estimation}
	Suppose that $\modelspace$ is a statistical model, which may be parametric, semiparametric, or nonparametric \citep{bickel1993efficient}. The data $\observation^*$ consists of $n$ independently and identically distributed (iid) random variables $O_i$, $i=1,\ldots,n$, following the true distribution $P_0 \in \modelspace$. We set $\coarsen$ to be the identity function so that $\observation=\observation^*$. We wish to estimate an aspect $\Psi(P_0) \in \real$ of $P_0$. Then, we can consider $\decisionspace$ to be a set of $\observationspace \rightarrow \real$ functions and the mean-squared error risk $R(d,P)= \expect_P[\{d(\observation) - \Psi(P)\}^2]$. Some specific examples of estimands include:
	\begin{enumerate}[label=\roman*)]
		\item Mean: $\Psi(P)=\expect_P[O_i]$;
		\item Cumulative distribution function at a point $o$: $\Psi(P)=\mathbb{P}_P(O_i \leq o)$;
		\item Correlation: with $O_i=(X_i,Y_i) \in \real^2$, $\Psi(P)=\expect_P[X_i Y_i] - \expect_P[X_i] \expect_P[Y_i]$.
	\end{enumerate}
\end{example}

\begin{example}[Predicting the expected number of novel categories to be observed in a new sample] \label{example: predict n new species}
	Suppose that $\modelspace$ consists of multinomial distributions with an unknown number of categories. Let an iid random sample of size $n$ be generated from the true multinomial distribution, so that $\observation^*$ is a multiset containing the number $X_k$ of observations in each category $k$. Suppose that only categories with nonzero occurrences are observed, so that $\observation$ is a left-truncated version of $\observation^*$. In other words, $\observation$ is the multiset $\coarsen(\observation^*)=\{X_k: X_k > 0\}$. Then, we may wish to predict the number of new categories that would be observed if a new sample of size $m$ were collected. This problem has been extensively studied in the literature, with applications in microbiome data, species taxonomic surveys, and assessment of vocabulary size, among other areas \citep[e.g.,][]{Shen2003,Bunge2014,Orlitsky2016}. This prediction problem can be formulated in our framework. For each $P \in \modelspace$, let $p_k$ ($k=1,\ldots,K_P$) be the probability of category $k$, and $\Psi(P)(\observation^*)$ be $\sum_{k=1}^{K_P} I(X_k=0) (1-(1-p_k)^m)$, the expected number of new observed categories given the current full data $\observation^*$. We consider $\decisionspace$ to be a set of $\observationspace \rightarrow \real$ functions and set the risk to be the mean-squared error, that is, $R(d,P)= \expect_P[\{d(\observation) - \Psi(P)(\observation^*)\}^2]$. This prediction problem is known to be intrinsically difficult when the future sample size $m$ is greater than the observed sample size $n$, and we might expect prior information to substantially improve prediction.
\end{example}

\begin{example}[Entropy estimation] \label{example: estimate entropy}
	Consider the same data-generating mechanism and observed data as in Example~\ref{example: predict n new species}. We may wish to estimate Shannon entropy \citep{Shannon1948} $\Psi(P) = -\sum_{k=1}^{K_P} p_k \log p_k$, a measure of diversity. We consider $\decisionspace$ to be a set of $\observationspace \rightarrow \real$ functions and set the risk to be the mean-squared error, that is, $R(d,P) = \expect_P[ \{d(\observation) - \Psi(P)\}^2]$. \cite{Jiao2015} proposed a rate-minimax estimator. Thus, in contrast to Example~\ref{example: predict n new species}, this is an example of a statistical problem with a satisfactory solution. For these problems, we might not expect prior information to substantially improve estimation.
\end{example}

We now define Gamma-minimaxity within our decision-theoretic framework. We assume that $\modelspace$ is equipped with the Borel $\sigma$-field $\mathcal{B}$ and let $\Pi$ denote the set of all probability distributions on the measurable space $(\modelspace,\mathcal{B})$. We also assume that, for any $d \in \decisionspace$ and any $\pi \in \Pi$, $P \mapsto R(d,P)$ is $\pi$-integrable. The Bayes risk corresponding to an estimator $d$ and a prior $\pi$ is defined as $r: (d,\pi) \mapsto \int R(d,P) \, \pi(\intd P)$. Let $\Gamma \subseteq \Pi$ be the set of priors such that all $\pi \in \Gamma$ are consistent with the available prior information. An estimator is called a $\Gamma$-minimax estimator if it is in the set
\begin{equation}
	\argmin_{d \in \decisionspace} \sup_{\pi \in \Gamma} r(d,\pi). \label{equation: define Gamma minimax}
\end{equation}
Throughout the rest of this paper, we assume the existence of this solution set and other solution sets to minimax problems, and that $\sup_{\pi \in \Gamma} r(d,\pi)$ is finite for any $d \in \decisionspace$.

In this paper, we consider the case in which $\Gamma$ is characterized by finitely many generalized moment conditions, that is,
$$\Gamma=\left\{ \pi \in \Pi: \Phi_k \in L^1(\pi), \int \Phi_k(P) \, \pi(\intd P) \leq c_k, k=1,\ldots,K \right\}$$
where each $\Phi_k: \modelspace \rightarrow \real$ is a prespecified function that extracts an aspect of a data-generating mechanism and $c_k \in \real$ is a prespecified constant.
The validity of our proposed template to find a $\Gamma$-minimax estimator in Section~\ref{section: increasing grid} does not require $\Gamma$ to take this form, but our proposed algorithms in Sections~\ref{section: SGDmax} and \ref{section: fictitious play} are computationally feasible for such constraints because these linear constraints lead to linear programs, which can be solved efficiently \citep[e.g.,][]{Jiang2020}. In principle, more general constraints can be handled by using suitable minimax problem solvers.
Such constraints were considered in \cite{Noubiap2001} and can represent a variety of forms of prior information. For example, with $\Phi_k=\pm \Psi^\kappa$ for some $\kappa \geq 1$, $\Gamma$ imposes bounds on prior moments of $\Psi(P)$; with $\Phi_k(P)=\pm \ind(\Psi(P) \in I)$ for some known interval $I$, $\Gamma$ imposes bounds on the prior probability of $\Psi(P)$ lying in $I$. Similar prior information on aspects of $P_0$ other than $\Psi(P_0)$ can also be represented. In addition, since an equality can be equivalently expressed by two inequalities, $\Gamma$ may also impose equality constraints on prior generalized moments.
Such information is commonly used to choose prior distributions in Bayesian settings \citep{Sarma2020}. Since we do not require specifying a parametric model or specifying an entire prior distribution for any finite-dimensional summary of $P_0$, specifying a set $\Gamma$ of prior distributions in the above form is no more difficult --- and often easier --- than specifying a single prior distribution, as would be required in a Bayesian approach.

\section{Proposed meta-learning algorithms to compute a $\Gamma$-minimax estimator} \label{section: algorithm}

Since both the model space $\modelspace$ and the estimator space $\decisionspace$ may be infinite, it is computationally infeasible to directly solve the minimax problem \eqref{equation: define Gamma minimax} defining a $\Gamma$-minimax estimator. Similarly to \cite{Noubiap2001}, our general strategy is to discretize $\modelspace$ and thus consider prior distributions with discrete supports. Once the supports of prior distributions are discrete, the optimization over prior distributions only involves finitely many parameters, namely the probability masses at support points, and thus is computationally possible. We will show that, when the grid is sufficiently fine, a solution to the discretized minimax problem is close to a solution to the original minimax problem.

Our proposed algorithm consists of two main steps. The first step is to discretize the model space $\modelspace$ and consider an approximating grid $\modelspace_\ell$ instead of the original complicated model space $\modelspace$. This discretization is illustrated in Fig.~\ref{fig: grid}. We will describe $\modelspace_\ell$ in more detail in Section~\ref{section: increasing grid}. In the second step, we consider a set $\Gamma_\ell$ of priors with support contained $\modelspace_\ell$ and compute a $\Gamma_\ell$-minimax estimator. We will describe two classes of algorithms to solve this discretized minimax problem in Sections~\ref{section: SGDmax} and \ref{section: fictitious play}, respectively.

\begin{figure}[bt!]
	\centering
	\includegraphics[scale=.35]{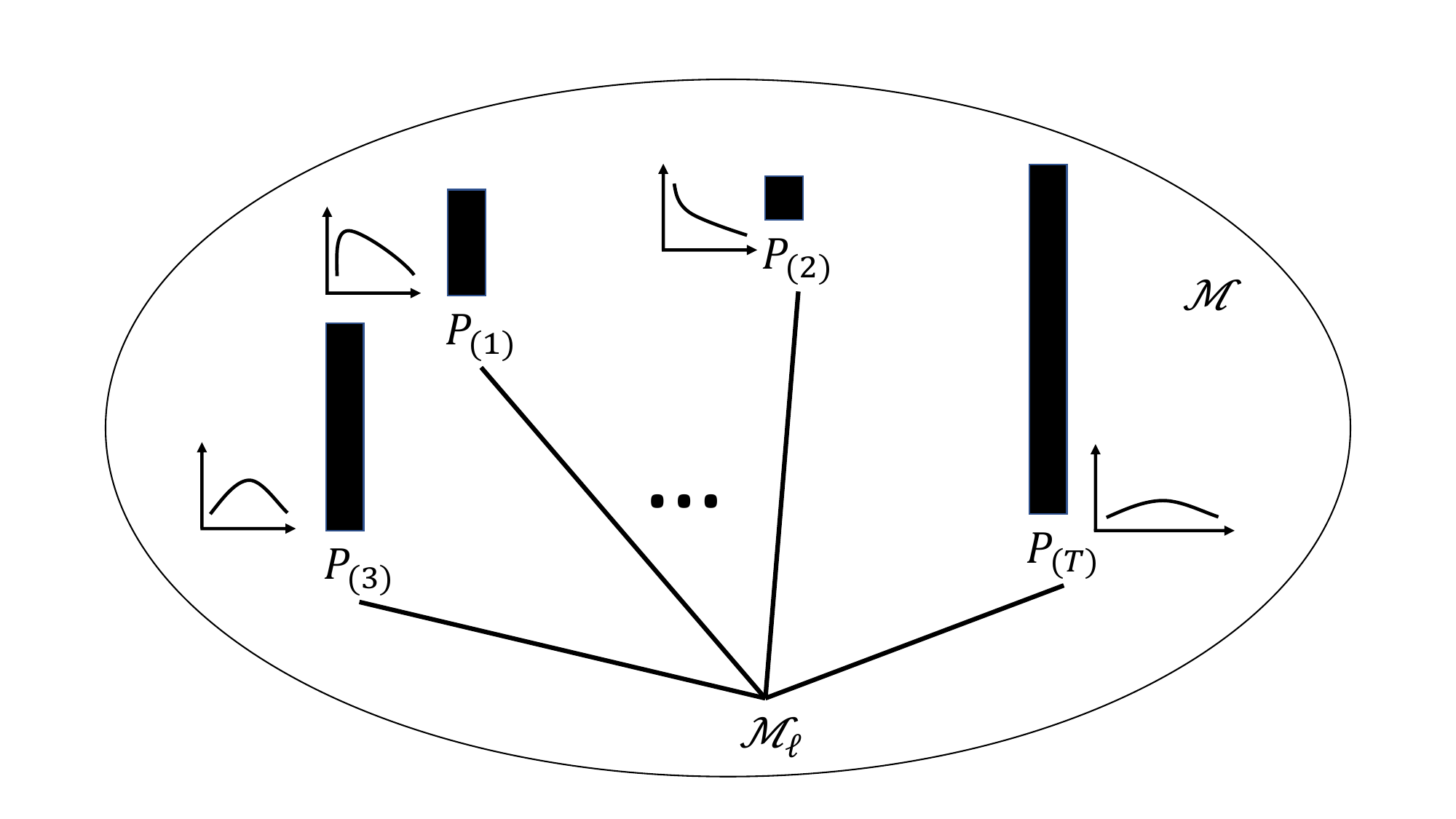}
	\caption{Illustration of grid $\modelspace_\ell=\{P_{(1)},P_{(2)},P_{(3)},\ldots,P_{(T)}\} \subseteq \modelspace$ approximating the entire model space $\modelspace$. Examples of densities of distributions $P_{(t)}$ ($t=1,\ldots,T$) in the grid are displayed. A prior distribution with support in $\modelspace_\ell$ is parameterized by the probability mass at each distribution $P_{(t)}$. An example of a prior distribution is displayed as black bars with their heights being proportional to the probability masses.}
	\label{fig: grid}
\end{figure}

\subsection{Grid-based approximation of $\Gamma$-minimax estimators} \label{section: increasing grid}

We first define the discretization of the model space $\modelspace$ that we will use. Let $\{\modelspace_\ell\}_{\ell=1}^\infty$ be an increasing sequence of finite subsets of $\modelspace$ such that $\bigcup_{\ell=1}^\infty \modelspace_\ell$ is dense in $\modelspace$. That is, $\{\modelspace_\ell\}_{\ell=1}^\infty$ is an increasingly fine grid over $\modelspace$. Since $\modelspace$ is separable, such an $\{\modelspace_\ell\}_{\ell=1}^\infty$ necessarily exists. Define
$$\Gamma_\ell := \{\pi \in \Gamma: \pi \text{ has support in } \modelspace_\ell\} \qquad \text{and} \qquad r_{\sup}(d,\Gamma') := \sup_{\pi \in \Gamma'} r(d,\pi)$$ for any $d \in \decisionspace$ and $\Gamma' \subseteq \Pi$.

Algorithm~\ref{algorithm: approximation with increasingly fine grid} describes how the grids $\modelspace_\ell$ are used to compute an approximately $\Gamma$-minimax estimator in our proposed algorithms. We will show that the approximation error decays to zero as $\ell$ grows to infinity.
Here and in the rest of the algorithms in the paper, for any real-valued function $f$, when we assign $\argmin_x f(x)$ or $\argmax_x f(x)$ to a variable, we arbitrarily pick a minimizer or maximizer if there are multiple optimizers. In practice, the user may stop the iteration at some $\ell$ and use a $\Gamma_\ell$-minimax estimator $d^*_\ell$ as the output estimator. We discuss the stopping criterion in more detail at the end of this section.

\begin{algorithm}
	\caption{Iteratively approximate a $\Gamma$-minimax estimator over an increasingly fine grid.}
	\label{algorithm: approximation with increasingly fine grid}
	\begin{algorithmic}[1]
		\For{$\ell=1,2,\ldots$}
		\State Construct a grid $\modelspace_{\ell} \subseteq \modelspace$ such that $\modelspace_{\ell-1} \subsetneq \modelspace_{\ell}$
		\State $d^*_\ell \gets \argmin_{d \in \decisionspace} \sup_{\pi \in \Gamma_\ell} r(d,\pi)$ \label{step: Gamma_l minimax}
		\EndFor
	\end{algorithmic}
\end{algorithm}

We note that the minimax problem in Line~\ref{step: Gamma_l minimax} of Algorithm~\ref{algorithm: approximation with increasingly fine grid} is nontrivial to solve, and therefore we propose two algorithms that can solve this minimax problem in Sections~\ref{section: SGDmax} and \ref{section: fictitious play}.

We assume that the following conditions hold throughout the rest of the paper.

\begin{condition} \label{condition: limit point}
	There exists a limit point $d^* \in \decisionspace$ of the sequence $\{d_\ell^*\}_{\ell=1}^\infty$.
\end{condition}

Condition~\ref{condition: limit point} holds if the sequence $\{d_\ell^*\}_{\ell=1}^\infty$ eventually falls in a compact set. For example, if $\decisionspace$ is a space of neural networks and we take $\varrho$ to be the Euclidean norm in the coefficient space, then we expect Condition~\ref{condition: limit point} to hold if all coefficients are restricted to fall in a bounded set, which is a common restriction in theoretical analyses of neural networks \citep[see, e.g.,][]{Goel2016,Zhang2016,Eckle2019} and often leads to desirable generalization bounds \citep[see, e.g.,][]{Bartlett1997,Bartlett2017,Neyshabur2017}.
Our theoretical results hold for any limit point $d^*$ in Condition~\ref{condition: limit point}, even if there is more than one of them.

\begin{condition} \label{condition: conditions on risk}
	The mapping $d \mapsto R(d,P)$ is continuous at $d^*$ for all $P \in \modelspace$.
\end{condition}

Condition~\ref{condition: conditions on risk} also often holds. For example, when parameterized using neural networks, all estimators are continuous functions of coefficients for common activation functions such as the sigmoid or the rectified linear unit (ReLU) \citep{Glorot2011} function, and therefore $d \mapsto R(d,P)$ is continuous everywhere.

We next present a sufficient condition to ensure that $d^*$ is $\Gamma$-minimax, so that $d^*_\ell$ is approximately $\Gamma$-minimax for sufficiently large $\ell$.

\begin{condition} \label{condition: subset of modelspace}
	We assume that there exists an increasing sequence $\{\Omega_\ell\}_{\ell=1}^\infty$ of subsets of $\modelspace$ such that
	\begin{enumerate}
		\item $\bigcup_{\ell=1}^\infty \Omega_\ell = \modelspace$; \label{condition: subset of modelspace, dense}
		\item for all $\ell=1,2,\ldots$ and all $d \in \decisionspace$, define $\tilde{\Gamma}_\ell := \{\pi \in \Gamma: \pi \text{ has support in } \Omega_\ell\}$ and $\tilde{\Gamma}_{i \mid \ell} := \{\pi \in \Gamma: \pi \text{ has support in } \modelspace_i \bigcap \Omega_\ell\}$. For any $\pi \in \tilde{\Gamma}_\ell$ with a finite support, there exists a sequence $\pi_i \in \tilde{\Gamma}_{i \mid \ell}$ such that $r(d,\pi_i) \rightarrow r(d,\pi)$ as $i \rightarrow \infty$. \label{condition: subset of modelspace, grid approximate subset}
	\end{enumerate}
\end{condition}

We note that, in contrast to $\modelspace_\ell$, $\Omega_\ell$ may be an infinite set. 
We may expect Condition~\ref{condition: subset of modelspace} to hold in many cases, especially when $P \mapsto R(d,P)$ is continuous at each $d \in \decisionspace$ and the grid $\modelspace_\ell$ contains a variety of distributions that are consistent with prior information represented by $\Gamma$. We illustrate this point with two counterexamples in Appendix~\ref{section: counterexample of subset of modelspace}. We will check the plausibility of Condition~\ref{condition: subset of modelspace} for Example~\ref{example: predict n new species} in our simulation and data analysis in Section~\ref{section: simulation n new species} for exemplar prior information; an almost identical argument shows the plausibility of Condition~\ref{condition: subset of modelspace} for Example~\ref{example: estimate entropy}.

We now present the theorem on $\Gamma$-minimaxity of $d^*$.

\begin{theorem}[Validity of grid-based approximation] \label{theorem: approximate Gamma-minimax on a grid}
	Under Conditions~\ref{condition: limit point}--\ref{condition: subset of modelspace}, $d^*$ is $\Gamma$-minimax and
	$$r_{\sup}(d_\ell^*,\Gamma_\ell) \nearrow \min_{d \in \decisionspace} r_{\sup}(d,\Gamma) \quad \text{as} \quad \ell \rightarrow \infty.$$
\end{theorem}

To prove Theorem~\ref{theorem: approximate Gamma-minimax on a grid}, we utilize a result in \cite{Pinelis2016} to establish that $r_{\sup}(d,\Gamma)$ can be approximated arbitrarily well by a discrete prior in $\Gamma$ for any $d \in \decisionspace$. This is a key ingredient in the proof of Lemma~\ref{lemma: approximate r_sup on subset}, which states that, for any $d\in\decisionspace$, $r_{\sup}(d,\tilde{\Gamma}_\ell)$ converges to $r_{\sup}(d,\Gamma)$. Then, we show that the sequence $\{r_{\sup}(d^*_\ell,\Gamma_{\ell})\}_{\ell=1}^\infty$ is nondecreasing and upper bounded by $\inf_{d \in \decisionspace} r_{\sup}(d,\Gamma)$, which is less than or equal to the $\Gamma$-maximal Bayes risk $r_{\sup}(d^*,\Gamma)$ of the limit point $d^*$ of $\{d_{\ell}^*\}_{\ell=1}^\infty$ in Condition~\ref{condition: limit point}. Therefore, $r_{\sup}(d^*_\ell,\Gamma_{\ell})$ converges to a limit. We finally use a contradiction argument to prove that this limit is greater than or equal to $r_{\sup}(d^*,\Gamma)$, which implies Theorem~\ref{theorem: approximate Gamma-minimax on a grid}.

We have the following corollary on the uniqueness of the $\Gamma$-minimax estimator and the convergence of $\{d^*_\ell\}_{\ell=1}^\infty$ for certain problems.

\begin{corollary}[Convergence of $\Gamma_\ell$-minimax estimator] \label{corollary: convexity; approximate Gamma-minimax on a grid}
	Suppose that $\decisionspace$ is a convex subset of a vector space, $d \mapsto R(d,P)$ is strictly convex for each $P \in \modelspace$, and $r_{\sup}(d,\Gamma)$ is attainable for each $d \in \decisionspace$ in the sense that, for all $d \in \decisionspace$, there exists a $\pi \in \Gamma$ such that $r(d,\pi)=r_{\sup}(d,\Gamma)$. Under Conditions~\ref{condition: limit point}--\ref{condition: subset of modelspace}, $d^*$ is the unique $\Gamma$-minimax estimator and
	$$d^*_\ell \rightarrow d^* \quad \text{as} \quad \ell \rightarrow \infty.$$
\end{corollary}

We prove Corollary~\ref{corollary: convexity; approximate Gamma-minimax on a grid} by establishing that $d \mapsto r_{\sup}(d,\Gamma)$ is strictly convex.

In practice, the user also needs to specify a stopping criterion for Algorithm~\ref{algorithm: approximation with increasingly fine grid}. In \cite{Noubiap2001}, the authors recommended computing or approximating
$r_{\sup}(d^*_\ell,\Gamma)$ and stop if $r_{\sup}(d^*_\ell,\Gamma)$ is sufficiently close to $r_{\sup}(d^*_\ell,\Gamma_{\ell})$. However, the procedure to approximate $r_{\sup}(d^*_\ell,\Gamma)$ in that work relies on the compactness of $\modelspace$, but we do not want to assume this condition because it may restrict the applicability of the method. Therefore, we propose to use the following alternative criterion: stop if $r_{\sup}(d^*_\ell,\Gamma_{\ell+1}) - r_{\sup}(d^*_\ell,\Gamma_\ell) \leq \epsilon$ for a prespecified tolerance level $\epsilon>0$. This criterion was proposed but not recommended in \cite{Noubiap2001} because it
does not guarantee that $r_{\sup}(d^*_\ell,\Gamma_\ell)$ is close to $r_{\sup}(d^*,\Gamma)$. For example, if $\modelspace_{\ell+1} \setminus \modelspace_\ell$ is small, it is even possible that $r_{\sup}(d^*_\ell,\Gamma_{\ell+1}) - r_{\sup}(d^*_\ell,\Gamma_\ell) = 0$, but $d^*_\ell$ is far from being $\Gamma$-minimax.
In contrast, we recommend this criterion for our proposed methods because we allow more flexibility in model specification, that is, $\modelspace$ need not be compact.
We discuss this issue in more detail in Section~\ref{section: increase grid with MCMC}.

We finally remark that $r_{\sup}(d,\Gamma_\ell)$ may be difficult to evaluate exactly. Since the risk is often an expectation, we recommend approximating $r_{\sup}(d,\Gamma_\ell)$ for any given $d$ via Monte Carlo as follows: first, estimate risks $R(d,P)$ for all $P \in \modelspace_\ell$ with a large number of Monte Carlo runs; second, estimate the corresponding least favorable prior $\pi_{d,\ell} \in \argmax_{\pi \in \Gamma_\ell} r(d,\pi)$ using the estimated risks; third, estimate the risks $R(d,P)$ ($P \in \modelspace_\ell$) again with independent Monte Carlo runs, and, finally, calculate $r(d,\pi_{d,\ell})$ with the estimated risks and the estimated least favorable prior. Using two independent estimates of the risk can remove the positive bias that would otherwise arise due to using the same data to estimate the risks and the least favorable prior.

\subsection{Computation of an estimator on a grid via stochastic gradient descent with max-oracle} \label{section: SGDmax}

In this section, we present methods to compute a $\Gamma_\ell$-minimax estimator, which corresponds to Line~\ref{step: Gamma_l minimax} in Algorithm~\ref{algorithm: approximation with increasingly fine grid}. Gradient descent with max-oracle (GDmax) and its stochastic variant (SGDmax), which were presented in \cite{Lin2020}, can be used to solve general minimax problems in Euclidean spaces. We focus on SGDmax in the main text and present GDmax in Appendix~\ref{section: GDmax}. To apply these algorithms to find a $\Gamma_\ell$-m~inimax estimator, we need to assume that $\decisionspace$ can be parameterized by a subset of a Euclidean space, that is, that for any $d \in \decisionspace$, there exists a real vector-valued coefficient
$\beta \in \real^D$
such that $d$ may be written as $d(\beta)$. For example, $\decisionspace$ may be a neural network class. More discussions on the parameterization of $\decisionspace$ can be found in Section~\ref{section: choice of estimator space}. In this section, in a slight abuse of notation, we define $R(\beta,P):=R(d(\beta),P)$, $r(\beta,\pi):=r(d(\beta),\pi)$ and $r_{\sup}(\beta,\Gamma_\ell) := r_{\sup}(d(\beta),\Gamma_\ell)$ for a coefficient $\beta \in \real^D$, a data-generating mechanism $P \in \modelspace$ and a prior $\pi \in \Gamma$. We assume that $\beta \mapsto R(\beta,P)$ is differentiable for all $P \in \modelspace$, and hence so is $\beta \mapsto r(\beta,\pi)$ for all $\pi \in \Gamma$.

It is often the case that $R(\beta,P)$ is expressed as an expectation. In this case, $R(\beta,P)$ may instead be approximated using Monte Carlo techniques. With $\xi$ being an exogenous source of randomness according to law $\Xi$, let $\hat{R}(\beta,P,\xi)$ be an unbiased approximation of $R(\beta,P)$ with $\expect[ \| \nabla_\beta \{\hat{R}(\beta,P,\xi) - R(\beta,P)\} \|^2 ] \leq \sigma^2 <\infty$, where $\| \cdot \|$ denotes the $\ell_2$-norm in Euclidean spaces. Let $\hat{r}(\beta,\pi,\xi) := \int \hat{R}(\beta,P,\xi) \, \pi(\intd P)$ for $\pi \in \Gamma_\ell$. In this case, SGDmax (Algorithm~\ref{algorithm: SGDmax}) may be used to find a (locally) $\Gamma_\ell$-minimax estimator. Note that Algorithm~\ref{algorithm: SGDmax} represents a generalization of the nested minimax AMC strategy in \cite{Luedtke2020} to $\Gamma_{\ell}$-minimax problems.

\begin{algorithm}
	\caption{Stochastic gradient descent with max-oracle (SGDmax) to compute a $\Gamma_\ell$-minimax estimator}
	\label{algorithm: SGDmax}
	\begin{algorithmic}[1]
		\State Initialize $\beta_{(0)} \in \real^D$. Set learning rate $\eta>0$, max-oracle accuracy $\zeta>0$ and batch size $J$.
		\For{$t=1,2,\ldots$}
		\State Stochastic maximization: use a stochastic procedure to find $\pi_{(t)} \in \Gamma_\ell$ such that $\expect[r(\beta_{(t-1)},\pi_{(t)})] \geq \max_{\pi \in \Gamma_\ell} r(\beta_{(t-1)},\pi) - \zeta$, where the expectation is over the randomness in stochastic maximization (e.g., variants of stochastic gradient ascent). \label{step: SGDmax stochastic max}
		\State Generate iid copies $\xi_1,\ldots,\xi_J$ of $\xi$.
		\State Stochastic gradient descent: $\beta_{(t)} \gets \beta_{(t-1)} - \frac{\eta}{J} \sum_{j=1}^{J} \nabla_\beta \hat{r}(\beta,\pi_{(t)},\xi_j) |_{\beta = \beta_{(t-1)}}$.
		\EndFor
	\end{algorithmic}
\end{algorithm}

We next present two conditions needed for the validity of Algorithm~\ref{algorithm: SGDmax}.

\begin{condition} \label{condition: uniform Lipschitz on R}
	For each $\ell=1,2,\ldots$ and all
	$\beta \in \real^D$,
	$\beta \mapsto R(\beta,P)$ is Lipschitz continuous with a universal Lipschitz constant $L_1$ independent of $P \in \modelspace_\ell$.
\end{condition}

Note that Condition~\ref{condition: uniform Lipschitz on R} differs from Condition~\ref{condition: conditions on risk} in that the former relies on the parameterization of $\decisionspace$ in a Euclidean space equipped with the Euclidean norm, while the latter may rely on a different metric on $\decisionspace$ such as an $L^2$-distance.

\begin{condition} \label{condition: uniform Lipschitz on R'}
	For each $\ell=1,2,\ldots$ and all
	$\beta \in \real^D$,
	$\nabla_\beta R(\beta,P)$ is bounded; $\beta \mapsto \nabla_\beta R(\beta,P)$ is Lipschitz continuous with a universal Lipschitz constant $L_2$ independent of $P \in \modelspace_\ell$.
\end{condition}

Under these conditions, using the results in \cite{Lin2020}, we can show that SGDmax yields an approximation to a local minimum of $\beta \mapsto r_{\sup}(\beta,\Gamma_\ell)$ when the algorithms' hyperparameters are suitably chosen. Before we formally present the theorem, we introduce some definitions related to the local optimality of a potentially nondifferentiable and nonconvex function. A real-valued function $f$ is called $q$-weakly convex if $x \mapsto f(x) + (q/2) \| x \|^2$ is convex ($q>0$). The Moreau envelope of a real-valued function $f$ with parameter $q>0$ is $f_q: x \mapsto \min_{x'} f(x') + \| x'-x \|^2/(2q)$. A point $x$ is an $\epsilon$-stationary point ($\epsilon \geq 0$) of a $q$-weakly convex function $f$ if $\| \nabla f_{1/(2q)}(x) \| \leq \epsilon$. Similarly, a random point $x$ is an $\epsilon$-stationary point ($\epsilon \geq 0$) of a $q$-weakly convex function $f$ in expectation if $\expect [\| \nabla f_{1/(2q)}(x) \|] \leq \epsilon$. If $x$ is an $\epsilon$-stationary point in expectation, we may conclude that it is an $\epsilon$-stationary point with high probability by Markov's inequality. Lemma~3.8 in \cite{Lin2020} shows that an $\epsilon$-stationary point of $f$ is close to a point $x'$ at which $f$ has at least one small subgradient for small $\epsilon$, so that $f(x')$ is close to a local minimum. In other words, if an algorithm outputs an estimator $\hat{d}=d(\hat{\beta})$ such that $\hat{\beta}$ is an $\epsilon$-stationary point of $\beta \mapsto r_{\sup}(\beta,\Gamma_\ell)$, then we know that $r_{\sup}(\hat{\beta},\Gamma_\ell)$ is close to a local minimum of $\beta \mapsto r_{\sup}(\beta,\Gamma_\ell)$.

We next present the validity result for Algorithm~\ref{algorithm: SGDmax}.

\begin{theorem}[Validity of SGDmax (Algorithm~\ref{algorithm: SGDmax})] \label{theorem: SGDmax convergence}
	Suppose that Conditions~\ref{condition: limit point}--\ref{condition: conditions on risk} and \ref{condition: uniform Lipschitz on R}--\ref{condition: uniform Lipschitz on R'} hold. Let $\epsilon>0$ be fixed and define $\Delta := (r_{\sup})_{1/(2 L_1)} (\beta_{(0)}) - \min_{\beta \in \real^D} (r_{\sup})_{1/(2 L_1)} (\beta)$, where we recall that $(r_{\sup})_{1/(2 L_1)}$ is the Moreau envelope of $r_{\sup}$ with parameter $1/(2 L_1)$. In Algorithm~\ref{algorithm: SGDmax}, with $\eta = \epsilon^2/[L_1 (L_2^2 + \sigma^2)]$, $\zeta = \epsilon^2/(24 L_1)$ and $J=1$, $\beta_{(t)}$ is an $\epsilon$-stationary point of $\beta \mapsto r_{\sup}(\beta,\Gamma_\ell)$ in expectation for $t = O(L_1 (L_2^2 + \sigma^2) \Delta/\epsilon^4)$, and is thus close to a local minimum of $\beta \mapsto r_{\sup}(\beta,\Gamma_\ell)$ with high probability.
\end{theorem}

The assumption that the batch size $J=1$ is purely for convenience since increasing $J$ corresponds to decreasing variance $\sigma^2$. To run Algorithm~\ref{algorithm: SGDmax} in practice, the user only needs to specify tuning parameters in Line~1 and all other constants in Theorem~\ref{theorem: SGDmax convergence} need not be known.
	In general, a small learning rate $\eta$, a stringent accuracy $\zeta$, and a large batch size $J$ make Algorithm~\ref{algorithm: SGDmax} likely to eventually reach an approximation of a local minimum of $\beta \mapsto r_{\sup}(\beta,\Gamma_\ell)$, but computation time might increase.
	Similar to most numeric optimization algorithms, fine-tuning is needed to achieve a balance between convergence guarantee and computation time, but a conservative choice of tuning parameters would typically result in convergence at the cost of computation time.

We note that Line~\ref{step: SGDmax stochastic max} in Algorithm~\ref{algorithm: SGDmax} may be inconvenient to implement because linear program solvers often do not use stochastic optimization. Therefore, we propose to use a convenient variant (Algorithm~\ref{algorithm: convenient SGDmax} in Appendix~\ref{section: GDmax}), where the stochastic maximization step (Line~\ref{step: SGDmax stochastic max} in Algorithm~\ref{algorithm: SGDmax}) is replaced by solving a linear program where the objective is approximated via Monte Carlo. This variant has similar validity under similar conditions.
We also note that the two uniform Lipschitz continuity conditions (\ref{condition: uniform Lipschitz on R} and \ref{condition: uniform Lipschitz on R'}) heavily rely on
the fact that $\modelspace_\ell$ is finite and the compactness of a set containing the coefficients. Nevertheless, the latter compactness restriction is common in theoretical analyses of neural networks  \citep[see, e.g.,][]{Goel2016,Zhang2016,Eckle2019}.
Moreover, these two conditions are sufficient conditions for the validity of the gradient-based methods, namely SGDmax, our variant of SGDmax and GDmax; a guarantee similar to these validity results might hold when two conditions are violated.

We finally remark that other algorithms similar to SGDmax can be applied, for example, (stochastic) gradient descent ascent with projection \citep{Lin2020}, (stochastic) mirror descent ascent, or accelerated (stochastic) mirror descent ascent \citep{Huang2021}. It is of future research interest to develop gradient-based methods to solve minimax problems with convergence guarantees under weaker conditions.

\subsection{Computation of an estimator on a grid via fictitious play} \label{section: fictitious play}

The algorithms in Section~\ref{section: SGDmax} may be convenient in many cases, but the requirements such as parameterization of the space $\decisionspace$ of estimators in a Euclidean space, differentiability of the risk function $R$ with respect to the coefficients $\beta$, and uniform Lipschitz continuity may be restrictive for certain problems. In this section, we propose an alternative algorithm, fictitious play, that avoids these requirements. We also present its convergence results.

\cite{Brown1951} introduced fictitious play as a means to find the value of a zero-sum game, that is, the optimal mixed strategy for both players and their expected gains. \cite{Robinson1951} then proved that fictitious play can be used to iteratively solve a two-player zero-sum game for a saddle point that is a pair of mixed strategies where both players have finitely many pure strategies. Our problem of finding a $\Gamma$-minimax estimator may also be viewed as a two-player zero-sum game where one player chooses a prior from $\Gamma$ and the other player chooses an estimator from $\decisionspace$. If we assume that, for the $\Gamma$-minimax problem at hand, the pair of both players' optimal strategies is a saddle point, which holds in many minimax problems \citep[e.g.,][]{V.Neumann1928,fan1953minimax,Sion1958}, then fictitious play may also be used to find a $\Gamma$-minimax estimator. Since $\Gamma$ may be too rich to allow for feasible implementation of fictitious play, we propose to use this algorithm to find a $\Gamma_\ell$-minimax estimator.

In the fictitious play algorithm in \cite{Robinson1951}, the two players take turns to play the best pure strategy against the mixture of the opponent's historic pure strategies, and the final output is a pair of mixtures of the two players' historic pure strategies. Since this algorithm aims to find minimax mixed strategies, we consider stochastic estimators. That is, consider the Borel $\sigma$-field $\mathcal{F}$ over $\decisionspace$ and let $\varPi$ denote the set of all probability distributions on the measurable space $(\decisionspace,\mathcal{F})$. We define $\overline{\decisionspace}$ to be the space of stochastic estimators with each element taking the following form: first draw an estimator from $\decisionspace$ according to a distribution $\varpi \in \varPi$ with an exogenous random mechanism and then use the estimator to obtain an estimate based on the data. Note that we may write any $\overline{d} \in \overline{\decisionspace}$ as $\overline{d}(\varpi)$ for some $\varpi \in \varPi$. We consider estimators in $\overline{\decisionspace}$ throughout this section, with the definition of $\Gamma$-minimaxity extended in the natural way, so that $\overline{d}^* = \overline{d}(\varpi^*) \in \overline{\decisionspace}$ is $\Gamma$-minimax if $r_{\sup}(\overline{d}^*,\Gamma) = \min_{\overline{d} \in \overline{\decisionspace}} r_{\sup}(\overline{d},\Gamma)$; we similarly extend all other definitions from Section~\ref{section: setup}. We assume that there exists $\pi^*_\ell \in \Gamma_\ell$ ($\ell=1,2,\ldots$) such that
\begin{equation}
	\label{equation: fictitious play minimax theorem}
	r ( \overline{d}^*,\pi^*_\ell ) = \sup_{\pi \in \Gamma_\ell} \inf_{\overline{d} \in \overline{\decisionspace}} r(\overline{d},\pi) = \inf_{\overline{d} \in \overline{\decisionspace}} \sup_{\pi \in \Gamma_\ell} r(\overline{d},\pi).
\end{equation}
In other words, $(\overline{d}^*,\pi^*_\ell)$ is a saddle point of $r$ in $\overline{\decisionspace} \times \Gamma_{\ell}$. Under this condition and the further conditions that $\decisionspace$ is convex and $d \mapsto R(d,P)$ is convex for all $P \in \modelspace$, it is possible to use a $\Gamma$-minimax estimator over the richer class $\overline{\decisionspace}$ of stochastic estimators to derive a $\Gamma$-minimax estimator over the original class $\decisionspace$. Indeed, for any $\overline{d}(\varpi) \in \overline{\decisionspace}$ and $P \in \modelspace$, by Jensen's inequality, $R(\overline{d}(\varpi),P) = \int R(d,P) \, \varpi(\intd d) \geq R(\underline{\overline{d}}(\varpi),P)$ where $\underline{\overline{d}}(\varpi) := \int d \, \varpi(\intd d) \in \decisionspace$ is the average of the stochastic estimator $\overline{d}(\varpi)$; that is, the risk of $\underline{\overline{d}}(\varpi)$ is never greater than that of $\overline{d}(\varpi)$. Therefore, we may use the fictitious play algorithm to compute $\overline{d}(\varpi^*_\ell)$ for each $\ell$ and further apply Algorithm~\ref{algorithm: approximation with increasingly fine grid} to compute $\overline{d}(\varpi^*)$. After that, we may take $\underline{\overline{d}}(\varpi^*)$ as the final output deterministic estimator.

Algorithm~\ref{algorithm: fictitious play} presents the fictitious play algorithm for finding a $\Gamma_\ell$-minimax estimator in $\overline{\decisionspace}$. Note that $\Gamma_\ell$ is convex, and hence $\pi$ always lies in $\Gamma_\ell$ throughout the iterations. In practice, we may initialize $\varpi$ as a point mass at an initial estimator in $\decisionspace$. In addition, similarly to \cite{Robinson1951}, we may replace Line~\ref{step: fictitious play find d} with $d^\dagger_{(t)} \gets \argmin_{d \in \decisionspace} r(d,\pi_{(t)})$, that is, minimizing the Bayes risk with the most recently updated prior rather than with the previous prior.

\begin{algorithm}
	\caption{Fictitious play to compute a $\Gamma_\ell$-minimax stochastic estimator}
	\label{algorithm: fictitious play}
	\begin{algorithmic}[1]
		\State Initialize $\varpi_{(0)} \in \varPi$ and $\pi_{(0)} \in \Gamma_\ell$.
		\For{t=1,2,$\ldots$}
		\State $\pi^\dagger_{(t)} \gets \argmax_{\pi \in \Gamma_\ell} r(\overline{d}(\varpi_{(t-1)}),\pi)$
		\State $\pi_{(t)} \gets \frac{t-1}{t} \pi_{(t-1)} + \frac{1}{t} \pi^\dagger_{(t)}$
		\State $d^\dagger_{(t)} \gets \argmin_{d \in \decisionspace} r(d,\pi_{(t-1)})$ \label{step: fictitious play find d}
		\State $\varpi_{(t)} \gets \frac{t-1}{t} \varpi_{(t-1)} + \frac{1}{t} \delta(d^\dagger_{(t)})$, where $\delta(d)$ denotes a point mass at $d \in \decisionspace$.
		\EndFor
	\end{algorithmic}
\end{algorithm}

We next present a convergence result for this algorithm.

\begin{theorem}[Validity of fictitious play (Algorithm~\ref{algorithm: fictitious play})] \label{theorem: fictitious play convergence}
	Assume that there exists a compact subset $\bar{\decisionspace}$ of $\decisionspace$ that contains all $d^\dagger_{(t)}$ ($t=1,2,\ldots$). Under Conditions~\ref{condition: limit point}--\ref{condition: conditions on risk}, it holds that
	$$r ( d^\dagger_{(t)},\pi_{(t-1)} ) \leq r ( \overline{d}(\varpi^*_\ell),\pi^*_\ell ) \leq r ( \overline{d}(\varpi_{(t-1)}),\pi^\dagger_{(t)} )$$
	for all $t$ and
	$$\lim_{t \rightarrow \infty} \left[ r ( \overline{d}(\varpi_{(t-1)}),\pi^\dagger_{(t)} ) - r ( d^\dagger_{(t)},\pi_{(t-1)} ) \right] = 0.$$
	Consequently, the $\Gamma_\ell$-maximal risk of $\overline{d}(\varpi_{(t)})$ converges to the $\Gamma_\ell$-minimax risk, that is,
	$$r_{\sup}(\overline{d}(\varpi_{(t-1)}),\Gamma_\ell) \rightarrow r_{\sup}(\overline{d}(\varpi_\ell^*),\Gamma_\ell) \quad \text{as} \quad t \rightarrow \infty.$$
\end{theorem}

\cite{Robinson1951} proved a similar case for two-player zero-sum games where each player has finitely many pure strategies. In contrast, in our problem, each player may have infinitely many pure strategies. A natural attempt to prove Theorem~\ref{theorem: fictitious play convergence} would be to consider finite covers of $\bar{\decisionspace}$ and $\Gamma_{\ell}$, that is, $\bar{\decisionspace} = \bigcup_{i=1}^I \decisionspace_i$ and $\Gamma_{\ell} = \bigcup_{j=1}^J \Pi_j$, such that the range of $r(d,\pi)$ in each $\decisionspace_i$ and $\Pi_j$ is small (say less than $\epsilon$), bin pure strategies into these subsets, and then apply the argument in \cite{Robinson1951} to these bins. The collection of $\decisionspace_i$ and $\Pi_j$ may be viewed as finitely many approximated pure strategies to $\Gamma_{\ell}$ and $\bar{\decisionspace}$ up to accuracy $\epsilon$, respectively. Unfortunately, we found that this approach fails. The problem arises because \cite{Robinson1951} inducted on $I$ and $J$, and, after each induction step, the corresponding upper bound becomes twice as large. Unlike the case with finitely many pure strategies that was considered in \cite{Brown1951} and \cite{Robinson1951}, as the desired approximation accuracy $\epsilon$ approaches zero, the numbers of approximated pure strategies, $I$ and $J$, may diverge to infinity, and so does the number of induction steps. Therefore, the resulting final upper bound is of order $2^{I+J} \epsilon$ and generally does not converge to zero as $\epsilon$ tends to zero. To overcome this challenge, we instead control the increase in the relevant upper bound after each induction step more carefully so that the final upper bound converges to zero as $\epsilon$ decreases to zero, despite the fact that $I$ and $J$ may diverge to infinity.

We remark that, because Line~\ref{step: fictitious play find d} of Algorithm~\ref{algorithm: fictitious play} typically involves another layer of iteration in addition to that over $t$, this algorithm will often be more computationally intensive than SGDmax. Nevertheless, Algorithm~\ref{algorithm: fictitious play} provides an approach to construct $\Gamma_{\ell}$-minimax estimators in cases where these other algorithms cannot be applied, for example, in settings where the risk is not differentiable in the parameters indexing the estimator or uniform Lipschitz conditions fail. In our numerical experiments, we have implemented this algorithm in the context of mean estimation (Appendix~\ref{section: simulation mean}). %

\section{Considerations in implementation} \label{section: considerations in implementation}

\subsection{Considerations when constructing the grid over the model space} \label{section: increase grid with MCMC}

By Theorem~\ref{theorem: approximate Gamma-minimax on a grid}, $r_{\sup}(d_\ell^*,\Gamma_\ell) \nearrow \min_{d \in \decisionspace} r_{\sup}(d,\Gamma)$ whenever Conditions~\ref{condition: limit point}--\ref{condition: subset of modelspace} hold and the increasing sequence $\{\modelspace_\ell\}_{\ell=1}^\infty$ is such that $\bigcup_{\ell=1}^\infty \modelspace_\ell$ is dense in $\modelspace$. Though this guarantee holds for all such sequences $\{\modelspace_\ell\}_{\ell=1}^\infty$, in practice, judiciously choosing this sequence of grids of distributions can lead to faster convergence. In particular, it is desirable that the least favorable prior $\Gamma_{\ell}$ puts mass on some of the distributions in $\modelspace_\ell\backslash\modelspace_{\ell-1}$ since, if this is not the case, then $d_{\ell}^*$ will be the same as $d_{\ell-1}^*$. %
While we may try to arrange for this to occur by adding many new points when enlarging $\modelspace_{\ell-1}$ to $\modelspace_\ell$, it may not be likely that any of these points will actually modify the least favorable prior unless they are carefully chosen.

To better address this issue, we propose to add grid points using a Markov chain Monte Carlo (MCMC) method. 
Our intuition is that, given an estimator $d$, the maximal Bayes risk is likely to significantly increase if we add distributions that (i) have a high risk for $d$, and (ii) are consistent with prior information so that there exists some prior such that these distributions lie in a high-probability region. We propose to use the MCMC algorithm to bias the selection of distributions in favor of those with the above characteristics. 
Let $\tau: \modelspace \rightarrow [0,\infty)$ denote a function such that $\tau(P) > \tau(P')$ if $P$ is more consistent with prior information than $P'$. For example, given a prior mean $\mu$ of some real-valued summary $\Psi(P)$ of $P$ and an interval $I$ that contains $\Psi(P)$ with prior probability at least 95\%, we may choose $\tau: P \mapsto \phi(\Psi(P))$, where $\phi$ is the density of a normal distribution that has mean $\mu$ and places 95\% of its probability mass in $I$. We call $\tau$ a pseudo-prior. Then, with the current estimator being $d$, we wish to select distributions $P$ for which $R(d,P) \tau(P)$ is large. We may use the Metropolis-Hastings-Green algorithm \citep{Metropolis1953,Hastings1970,Green1995} to draw samples from a density proportional to $P \mapsto R(d,P) \tau(P)$. We then let $\modelspace_{\ell}$ be equal to the union of $\modelspace_{\ell-1}$ and the set containing all unique distributions in this sample.

Details of the proposed scheme are provided in Algorithm~\ref{algorithm: MCMC-like}. To use this proposed algorithm, we rely on it being possible to define a sequence of parametric models $\{\tilde{\Omega}_{\ell}\}_{\ell=1}^\infty$ such that $\tilde{\modelspace}:=\cup_{\ell=1}^\infty \tilde{\Omega}_{\ell}$ is dense in $\modelspace_{\ell}$---this is possible in many interesting examples \citep[see, e.g.,][]{Chen2007}. When combined with separability of $\modelspace$, this condition enables the definition of an increasing sequence of grids of distributions $\{\modelspace_{\ell}\}_{\ell=1}^\infty$ such that, for each $\ell$, $\modelspace_{\ell}\subseteq \tilde{\modelspace}$.

\begin{algorithm}
	\caption{MCMC algorithm to construct $\modelspace_\ell$}
	\label{algorithm: MCMC-like}
	\begin{algorithmic}[1]
		\Require Previous grid $\modelspace_{\ell-1}$, current estimator $d^*_{\ell-1}$ and number $T$ of iterations. We define $\modelspace_{-1} := \emptyset$. An initial estimator $d^*_0$ must be available if $\ell=1$.
		\State Initialize $P_{(0)} \in \tilde{\modelspace}$.
		\For{$t=1,2,\ldots,T$}
		\State Propose a distribution $P' \in \tilde{\modelspace}$ from $P_{(t-1)}$ \label{step: MCMC propose}
		\State Calculate the MCMC acceptance probability $p_\mathrm{accept}$ of $P'$ for target density $P \mapsto R(d^*_{\ell-1},P) \tau(P)$
		\State With probability $p_\mathrm{accept}$, accept $P'$ and $P_{(t)} \gets P'$
		\If{$P'$ is not accepted}
		\State $P_{(t)} \gets P_{(t-1)}$
		\EndIf
		\EndFor
		\State $\modelspace_\ell \gets \text{unique elements of the multiset } \modelspace_{\ell - 1} \bigcup \{P_{(1)}, P_{(2)}, \ldots, P_{(T)}\}$
	\end{algorithmic}
\end{algorithm}

The following theorem on distributional convergence follows from that for the Metropolis-Hastings-Green algorithm \citep[see Section~3.2 and 3.3 of ][]{Green1995}.

\begin{theorem}[Validity of MCMC algorithm (Algorithm~\ref{algorithm: MCMC-like})] \label{theorem: MCMC}
	Suppose that $P \mapsto R(d^*_{\ell-1},P) \tau(P)$ is bounded and integrable with respect to some measure $\mu$ on $\tilde{\modelspace}$ and let $\mathscr{L}$ denote the probability law on $\tilde{\modelspace}$ whose density function with respect to $\mu$ is proportional to this function. Suppose that the MCMC is constructed such that the Markov chain is irreducible and aperiodic. Then, $P_{(t)}$ converges weakly to $\mathscr{L}$ as $t \rightarrow \infty$.
\end{theorem}

Therefore, if $\mathscr{L}$ corresponds to a continuous distribution with nonzero density over the parameter space of $\tilde{\modelspace}$, then Theorem~\ref{theorem: MCMC} implies that $\bigcup_{\ell=1}^\infty \modelspace_\ell$ is dense in $\modelspace$, as required by Algorithm~\ref{algorithm: approximation with increasingly fine grid}.

Implementing Algorithm~\ref{algorithm: MCMC-like} relies on the user making several decisions. These decisions include the choice of the pseudo-prior $\tau$ and the technique used to approximate the risk $R(d,P)$ to a reasonable accuracy. Fortunately, regardless of the decisions made, Theorem~\ref{theorem: approximate Gamma-minimax on a grid} suggests that $r_{\sup}(d_\ell^*,\Gamma_\ell) \nearrow \min_{d \in \decisionspace} r_{\sup}(d,\Gamma)$ for a wide range of sequences $\{\modelspace_{\ell}\}_{\ell=1}^\infty$. Indeed, all that theorem requires on this sequence is that the grid $\modelspace_{\ell}$ becomes arbitrarily fine as $\ell$ increases. Though the final decisions made are not important when $\ell$ is large, we still comment briefly on the decisions that we have made in our experiments, First, we have found it effective to approximate $R(d,P)$ via a large number of Monte Carlo draws. Second, in a variety of settings, we have also identified, via numerical experiments, candidate pseudo-priors that balance high risk and consistency with prior information (see Sections~\ref{section: simulation n new species} and \ref{section: simulation entropy} for details).

\subsection{Considerations when choosing the space of estimators} \label{section: choice of estimator space}

It is desirable to consider a rich space $\decisionspace_0$ of estimators to obtain an estimator with low maximal Bayes risk, and thus good general performance. However, to make numerically constructing these estimators computationally feasible, we usually have to consider a restricted space $\decisionspace$ of estimators.
This approximation is justified because, if estimators in $\decisionspace$
can approximate the $Gamma$-minimax estimator in $\decisionspace_0$
well, then we expect the resulting excess maximal Bayes risk is small.

Feedforward neural networks (or neural networks for short) are natural options for the space of estimators because of their universal approximation property \citep[e.g.,][]{Hornik1991,Csaji2001,Hanin2017,Kidger2020}. However, training commonly used neural networks can be computationally intensive. Moreover, a space of neural networks is typically nonconvex, and hence it may be difficult to find a global minimizer of the maximal Bayes risk even if the risk is convex in the estimator. Therefore, the learned estimator might not perform well.

To help overcome this challenge, we advocate for utilizing available statistical knowledge when designing the space of estimators. We call estimators that take this form \emph{statistical knowledge networks}. In particular, if a simple estimator is already available, we propose to use neural networks with such an estimator as a node connected to the output node. An example of such an architecture is presented in Fig.~\ref{figure: example nn}. In this sample architecture, each node is an activation function such as the sigmoid or the rectified linear unit (ReLU) \citep{Glorot2011} function applied to an affine transformation of the vector containing the ancestors of the node. The only exception is the output node, which is again an affine transformation of its ancestors but uses the identity activation function. When training the neural network, we may initialize the affine transformation in the output layer to only give weight to the simple estimator.
Under this approach, the space of estimators is a set of perturbations of an existing simple estimator. Although we may still face the challenge of nonconvexity and local optimality, we can at least expect to improve the initial simple estimator.

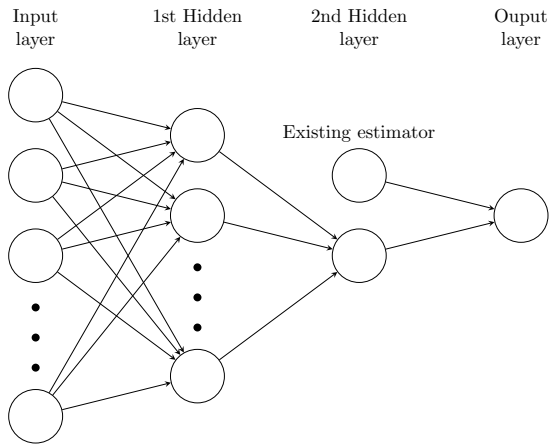
\begin{figure}[bt!]
	\centering
	\resizebox{!}{2.3in}{%
		\begin{tikzpicture}[x=1.5cm, y=1.5cm, >=stealth]
			\foreach \m [count=\y] in {1,2,3,missing,4}
			\node [every neuron/.try, neuron \m/.try] (input-\m) at (0,2.5-\y) {};
			
			\foreach \m [count=\y] in {1,2,missing,3}
			\node [every neuron/.try, neuron \m/.try ] (hidden1-\m) at (2,2-\y) {};
			
			\foreach \m [count=\y] in {1,2}
			\node [every neuron/.try, neuron \m/.try ] (hidden2-\m) at (4,1.5-\y) {};
			
			\foreach \m [count=\y] in {1}
			\node [every neuron/.try, neuron \m/.try ] (output-\m) at (6,0) {};

			\node [above] at (hidden2-1.north) {Existing estimator};
			
			\foreach \i in {1,...,4}
			\foreach \j in {1,...,3}
			\draw [->] (input-\i) -- (hidden1-\j);
			
			\foreach \i in {1,...,3}
			\foreach \j in {2}
			\draw [->] (hidden1-\i) -- (hidden2-\j);
			
			\foreach \i in {1,...,2}
			\foreach \j in {1}
			\draw [->] (hidden2-\i) -- (output-\j);
			
			\foreach \l [count=\x from 0] in {Input, 1st Hidden, 2nd Hidden, Ouput}
			\node [align=center, above] at (\x*2,2) {\l \\ layer};
		\end{tikzpicture}%
	}
	\caption{Example of neural network estimator architecture utilizing an existing estimator. The arrows from the input nodes to the existing estimator are omitted from this graph.}
	\label{figure: example nn}
\end{figure}

In the simulation we describe in Appendix~\ref{section: simulation mean}, we compared the empirical performance of several spaces of estimators. This simulation concerns the simple problem of estimating the mean of a true distribution whose support has known bounds (Example~\ref{example: estimation}), and the existing simple estimator we use in the statistical neural network is the sample mean. Fig.~\ref{figure: mean parameter Bayes risk for least favorable prior} presents the trajectory of estimated Bayes risks. As shown in subfigures~(b)--(d), using the statistical knowledge network, the estimator is almost $\Gamma$-minimax after a few iterations; on the other hand, it took about 1000 iterations for the feedforward neural network to reach an approximately $\Gamma$-minimax estimator. Therefore, in this simple problem where the true $\Gamma$-minimax estimator is a shifted and scaled sample mean, statistical knowledge substantially reduced the number of iterations required to obtain an approximately $\Gamma$-minimax estimator. For more complicated problems, we expect statistical knowledge to further help improve the performance of the computed estimator.

\begin{figure}[bt!]
	\centering
	\includegraphics[scale=0.4]{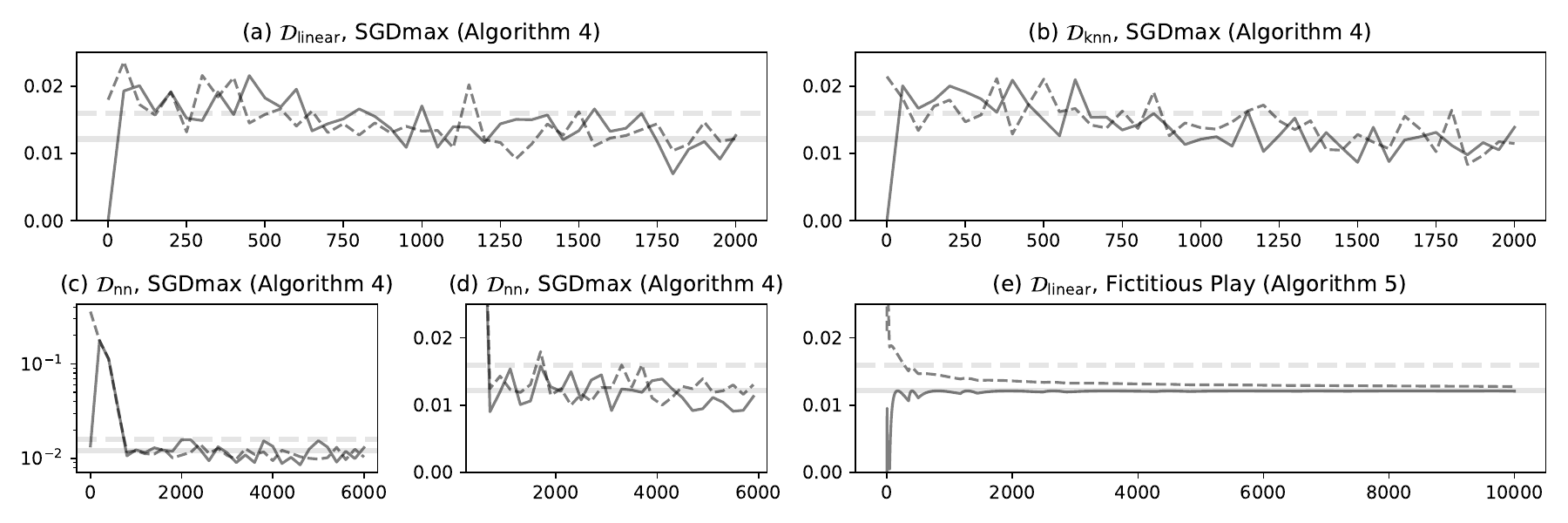}
	\caption{Estimated Bayes risks of the estimator over iterations when computing a $\Gamma_1$-minimax estimator. The lines are the current Bayes risks (y-axis) over iterations (x-axis) (unbiased estimates with 50 Monte Carlo runs for Algorithm~\ref{algorithm: convenient SGDmax}; exact values for Algorithm~\ref{algorithm: fictitious play}). The solid lines are the Bayes risks after an update in the estimator to decrease the Bayes risk. The dashed lines are the Bayes risks after an update in the prior to increase the Bayes risk. The two horizontal lines are the Bayes risk of the sample mean (dashed) and $d^*$ (solid), respectively, for $\pi^*$. For ease of visualization, in subfigures~(a) and (b), the Bayes risks are plotted every 50 iterations; in subfigures~(c) and (d), the Bayes risks are plotted every 200 iterations; subfigure~(d) contains the part in subfigure~(c) after 500 iterations.}
	\label{figure: mean parameter Bayes risk for least favorable prior}
\end{figure}

We note that we might overcome the challenge of nonconvexity and local optimality by using an extreme learning machine (ELM) \citep{Huang2006} to parameterize the estimator. ELMs are neural networks for which the weights in hidden nodes are randomly generated and are held fixed, and only the weights in the output layer are trained. Thus, the space of ELMs with a fixed architecture and fixed hidden layer weights is convex. Like traditional neural networks, ELMs have the universal approximation property \citep{Huang2006universal}. In addition, Corollary~\ref{corollary: convexity; approximate Gamma-minimax on a grid} may be applied to an\ ELM so that the $\Gamma_\ell$-minimax estimator may converge to the $\Gamma$-minimax estimator. As for traditional neural networks, we may incorporate knowledge of existing statistical estimators into an ELM.

We finally remark that, besides computational intensity when constructing (i.e., learning) a $\Gamma$-minimax estimator, another important factor to be considered when choosing $\decisionspace$ is the computational intensity to evaluate the learned estimator at the observed dataset. This is another reason for our choosing neural networks or ELMs as the space of estimators. Indeed, existing software packages \citep[e.g.,][]{pytorch} make it easy to leverage graphics processing units to efficiently evaluate the output of neural networks for any given input. Therefore, if the existing estimator being used is not too difficult to compute, then estimators parameterized using similar architectures to that displayed in Figure~\ref{figure: example nn} will be able to be computed efficiently in practice. This efficiency may be especially important in settings where the estimator will be applied to many datasets, so that the cost of learning the estimator is amortized and the main computational expense is evaluating the learned estimator.

\section{Simulations and data analyses} \label{section: simulation}

We illustrate our methods in Examples~\ref{example: estimation}--\ref{example: estimate entropy}. A toy example of Example~\ref{example: estimation} is presented in Appendix~\ref{section: simulation mean}. We focus on the more complex Examples~\ref{example: predict n new species} and \ref{example: estimate entropy} in this section.

\subsection{Prediction of the expected number of new categories} \label{section: simulation n new species}

We apply our proposed method to Example~\ref{example: predict n new species}. In the simulation, we set the true population to be an infinite population with the same categories and same proportions as the sample studied in \cite{Miller1989}, which consists of 1088 observations in 188 categories. This setting is the same as the simulation setting in \cite{Shen2003}. We set the sample size to be $n=100$ and the size of the new sample to be $m=200$. In this setting, the expected number of new categories in the new sample unconditionally on the observed sample, namely $\Phi(P_0):=\expect_{P_0}[\Psi(P_0)(\observation^*)]$, can be analytically computed and equals $48.02$. We note that this quantity can also be computed via simulation: (i) sample $n$ and $m$ individuals with replacement from the dataset in \cite{Miller1989}, (ii) count the number of new categories in the second sample, and (iii) repeat steps~(i) and (ii) many times and compute the average.

It is well known that this prediction problem is difficult when $m>n$, and we run this simulation to investigate the potential gain from leveraging prior information by computing a Gamma-minimax estimator for such difficult or even ill-posed problems. We consider three sets of prior information:
\begin{enumerate}
	\item strongly informative: prior mean of $\Phi(P)$ in $[45,50]$, $\geq 95\%$ prior probability that $\Phi(P)$ lies in $[40,55]$;
	\item weakly informative: prior mean of $\Phi(P)$ in $[40,55]$, $\geq 95\%$ prior probability that $\Phi(P)$ lies in $[30,65]$; and
	\item almost noninformative: prior mean of $\Phi(P)$ in $[35,60]$, $\geq 95\%$ prior probability that $\Phi(P)$ lies in $[20,75]$.
\end{enumerate}
We note that a traditional Bayesian approach would require specifying a prior on $\modelspace$, including the total number of categories and the proportion of each category, which may be difficult in practice.

We check the plausibility of Condition~\ref{condition: subset of modelspace} in this context. We take the strongly informative prior information as an example. Take $\Omega_\ell$ to be the collection of multinomial distributions with at most $\ell$ categories. It is obvious that $\bigcup_{\ell=1}^\infty \Omega_\ell=\modelspace$. Let $d \in \decisionspace$ be fixed and $\pi \in \tilde{\Gamma}_\ell$ be a fixed prior with finite support, that is, $\pi = \sum_{j=1}^J q_j \delta(Q_j)$ where $\delta(\cdot)$ denotes the point mass distribution, $Q_j \in \Omega_\ell$, $q_j>0$ and $\sum_{j=1}^J q_j=1$. Let $\epsilon>0$ be an arbitrary small number such that $\sum_{j=1}^J q_j \Phi(Q_j) \leq 50 - \epsilon$ or $\sum_{j=1}^J q_j \Phi(Q_j) \geq 45 + \epsilon$. Since $\bigcup_{\ell=1}^\infty \modelspace_\ell$ is dense in $\modelspace$ and $\Phi$ is continuous, there exists a sufficiently large $i$ such that, for every distribution $Q_j$, there exists $P_j \in \modelspace_i \cap \Omega_\ell$ satisfying the following:
\begin{itemize}
	\item $|\Phi(P_j)-\Phi(Q_j)| \leq \epsilon$;
	\item if $\Phi(Q_j) \in [40,55]$, then $\Phi(P_j) \in [40,55]$;
	\item $|R(d,P_j)-R(d,Q_j)| \leq \epsilon$.
\end{itemize}
Take $\pi_i$ to be $\sum_{j=1}^J q_j \delta(P_j)$. Then it is easy to verify that $|\sum_{j=1}^J q_j \Phi(P_j) - \sum_{j=1}^J q_j \Phi(Q_j)| \leq \epsilon$ and thus $\sum_{j=1}^J q_j \Phi(P_j) \in [45,50]$; moreover, $\Phi(Q_j) \in [40,55]$ implies that $\Phi(P_j) \in [40,55]$ and therefore $\sum_{j=1}^J q_j \ind( \Phi(P_j) \in [40,55]) \geq \sum_{j=1}^J q_j \ind( \Phi(Q_j) \in [40,55]) \geq 95\%$. Thus, $\pi_i \in \tilde{\Gamma}_{i \mid \ell}$. Moreover, $|r(d,\pi) - r(d,\pi_i)| \leq \epsilon$. Therefore, $r(d,\pi_i) \rightarrow r(d,\pi)$ as $i \rightarrow \infty$ and Condition~\ref{condition: subset of modelspace} holds.

We design the architecture of the neural network estimator as in Fig.~\ref{figure: nn n new species}. We choose two existing estimators (referred to as the OSW and SCL estimators, respectively) proposed by \cite{Orlitsky2016} and \cite{Shen2003} as human knowledge inputs to the architecture.
We use the ReLU activation function. There are 50 hidden nodes in the first hidden layer. %
We initialize the neural network that we train to output the average of these two existing estimators.

\begin{figure}[bt!]
	\centering
	\resizebox{!}{2.3in}{%
		\begin{tikzpicture}[x=1.5cm, y=1.5cm, >=stealth]
			\foreach \m [count=\y] in {1,missing,2}
			\node [every neuron/.try, neuron \m/.try] (input-\m) at (0,0.5-\y) {};
			
			\node [every neuron/.try, neuron 1/.try] (OSW) at (0,2) {};
			\node [every neuron/.try, neuron 1/.try] (SCL) at (0,1) {};
			
			\foreach \m [count=\y] in {1,missing,2}
			\node [every neuron/.try, neuron \m/.try ] (hidden1-\m) at (2,0.5-\y) {};
			
			\foreach \m [count=\y] in {1}
			\node [every neuron/.try, neuron \m/.try ] (hidden2-\m) at (4,-1) {};
			
			\foreach \m [count=\y] in {1}
			\node [every neuron/.try, neuron \m/.try ] (output-\m) at (6,0) {};
			
			\foreach \l [count=\i] in {1,n}
			\draw [<-] (input-\i) -- ++(-1,0)
			node [above, midway] {$X_\l$};
			
			\node [left] at (OSW.west) {OSW};
			\node [left] at (SCL.west) {SCL};

			\foreach \i in {1,...,2}
			\foreach \j in {1,...,2}
			\draw [->] (input-\i) -- (hidden1-\j);
			
			\foreach \j in {1,...,2}
			\draw [->] (OSW) -- (hidden1-\j);
			
			\foreach \j in {1,...,2}
			\draw [->] (SCL) -- (hidden1-\j);
			
			\foreach \i in {1,...,2}
			\foreach \j in {1}
			\draw [->] (hidden1-\i) -- (hidden2-\j);
			
			\draw [->] (OSW) -- (hidden2-1);
			\draw [->] (SCL) -- (hidden2-1);
			
			\draw [->] (OSW) -- (output-1);
			\draw [->] (SCL) -- (output-1);
			\draw [->] (hidden2-1) -- (output-1);
			
			\foreach \l [count=\x from 0] in {Input, 1st Hidden, 2nd Hidden, Ouput}
			\node [align=center, above] at (\x*2,2.5) {\l \\ layer};
		\end{tikzpicture}%
	}
	\caption{Architecture of the neural network estimator of the expected number of new categories. $X_k$: number of categories with $k$ observations; OSW: the estimator proposed in \cite{Orlitsky2016}; SCL: the estimator proposed in \cite{Shen2003}. The arrows from data $(X_1,\ldots,X_n)$ to the OSW and SCL estimators are omitted from this graph.}
	\label{figure: nn n new species}
\end{figure}

We use Algorithm~\ref{algorithm: MCMC-like} to construct $\modelspace_\ell$. There are 2000 grid points in $\modelspace_1$, and we add 1000 grid points each time we enlarge the grid. When generating $\modelspace_1$, we chose the starting point to be a distribution $P_{(0)}$ with 146 categories and $\Phi(P_{(0)}) = 49.9$. The choice of this starting point $P_{(0)}$ was quite arbitrary. We first generated a sample from $P_0$ and treated it as data from a pilot study. We then came up with a distribution $P_{(0)}$ such that five random samples generated from $P_{(0)}$ all appear qualitatively similar to the pilot data. In practice, this starting point can be chosen based on prior knowledge. Our chosen grid sizes for Algorithm~\ref{algorithm: MCMC-like} were quite arbitrary. For $\modelspace_1$, the generated distributions $P_{(t)}$ appear similar for all $t$, and thus the initial grid size 2000 and the increment size 1000 appeared sufficient. Smaller grid sizes would simply lead to more iterations in Algorithm~\ref{algorithm: approximation with increasingly fine grid}, which effectively increases the grid size.
We selected the log pseudo-prior as a weighted sum of two log density functions: (i) a normal distribution with the mean being the midpoint of the interval constraint on the prior mean of $\Phi(P)$ and central 95\% probability interval being the interval with at least 95\% prior probability, (ii) a negative-binomial distribution of the total number of categories with success probability $0.995$ and $2$ failures until the Bernoulli trial is stopped so that the mode and the variance are approximately $200$ and $8 \times 10^4$, respectively. These log-densities are provided weights 30 and 10, respectively. We selected the weights based on the empirical observation that distributions with only a few categories tend to have high risks, but these distributions are relatively inconsistent with prior information and may well be given almost negligible probability weight in a computed least favorable prior, thus contributing little to computing a $\Gamma$-minimax estimator. We chose the aforementioned weights so that Algorithm~\ref{algorithm: MCMC-like} can explore a fairly large range of distributions and does not generate too many distributions with too few categories.

We use Algorithm~\ref{algorithm: convenient SGDmax} with learning rate $\eta=0.005$ and batch size $J=30$ to compute $\Gamma_\ell$-minimax estimators. The number of iterations is 4,000 for $\Gamma_1$ and 200 for $\Gamma_\ell$ ($\ell>1$). The stopping criterion in Algorithm~\ref{algorithm: approximation with increasingly fine grid} is that the estimated maximal Bayes risk with 2000 Monte Carlo runs does not relatively increase by more than 2\% or absolutely increase by more than 0.0001.
We chose the aforementioned tuning parameters based on the prior belief that at least one of OSW and SCL estimators should perform reasonably well, but the performance of SGDmax (Algorithm~\ref{algorithm: convenient SGDmax}) and Algorithm~\ref{algorithm: MCMC-like} might be sensitive to tuning parameters. Thus, the network we used is neither deep nor wide. We chose a moderately small learning rate and a large number of iterations for SGDmax. Our chosen learning rate and chosen number of iterations led to a trajectory of estimated Bayes risks that approximately reached a plateau with small fluctuations, suggesting that the obtained estimator is approximately $\Gamma_1$-minimax (see Fig.~\ref{figure: n new species Bayes risk and risk}). In practice, such trajectory plots may help tune the learning rate and the number of iterations.

We also run additional simulations to investigate the sensitivity of our methods to tuning parameter selections. We present these simulations in Appendix~\ref{section: simulation n new species sensitivity}. The results suggest that our methods may be insensitive to tuning parameter selections.

We examine the performance of the OSW estimator, the SCL estimator and our trained $\Gamma$-minimax estimator by comparing their risks under our set data-generating mechanism computed with 20000 Monte Carlo runs. We also compare their Bayes risks under the computed prior from Algorithm~\ref{algorithm: convenient SGDmax} using the last and finest grid in the computation with 20000 Monte Carlo runs. We present the results in Table~\ref{table: new species result}. In this simulation experiment, our $\Gamma$-minimax estimator substantially reduces the risk compared to two existing estimators. The $\Gamma$-minimax estimator also has the lowest Bayes risk in all cases. Therefore, incorporating fairly informative prior knowledge into the estimator may lead to a significant improvement in predicting the number of new categories. We expect similar substantial improvement for difficult or even ill-posed statistical problems by incorporating prior knowledge.

\begin{table}[bt!]
	\centering
	\caption{Risks and Bayes risks of estimators. $R(d,P_0)$: risk of the estimator under the true data-generating mechanism $P_0$. $r(d,\hat{\pi}^*)$: Bayes risk under prior $\hat{\pi}^*$, the computed prior from Algorithm~\ref{algorithm: convenient SGDmax} in the last and finest grid in the computation.}
	\label{table: new species result}
	\begin{tabular}{l|l|r|r}
		\hline
		Strength of prior & Estimator & $R(d,P_0)$ & $r(d,\hat{\pi}^*)$ \\
		\hline
		strong & OSW & 265 & 303 \\
		& SCL & 146 & 159 \\
		& $\Gamma$-minimax & 18 & 35 \\
		\rule{0pt}{4ex} weak & OSW & 265 & 328 \\
		& SCL & 146 & 184 \\
		& $\Gamma$-minimax & 17 & 61 \\
		\rule{0pt}{4ex} almost none & OSW & 265 & 293 \\
		& SCL & 146 & 124 \\
		& $\Gamma$-minimax & 24 & 81 \\
		\hline
	\end{tabular}
\end{table}

Fig.~\ref{figure: n new species Bayes risk and risk} presents the unbiased estimator of Bayes risks over iterations when computing a $\Gamma_1$-minimax estimator. The Bayes risks appear to have a decreasing trend and to approach a liming value. Over iterations, the Bayes risks decrease by a considerable amount. The limiting value of the Bayes risks appears to be slightly higher than the risk of the computed $\Gamma$-minimax estimator under $P_0$. This might indicate that $P_0$ is not an extreme distribution that yields a high risk.

\begin{figure}[bt!]
	\centering
	\includegraphics[scale=0.4]{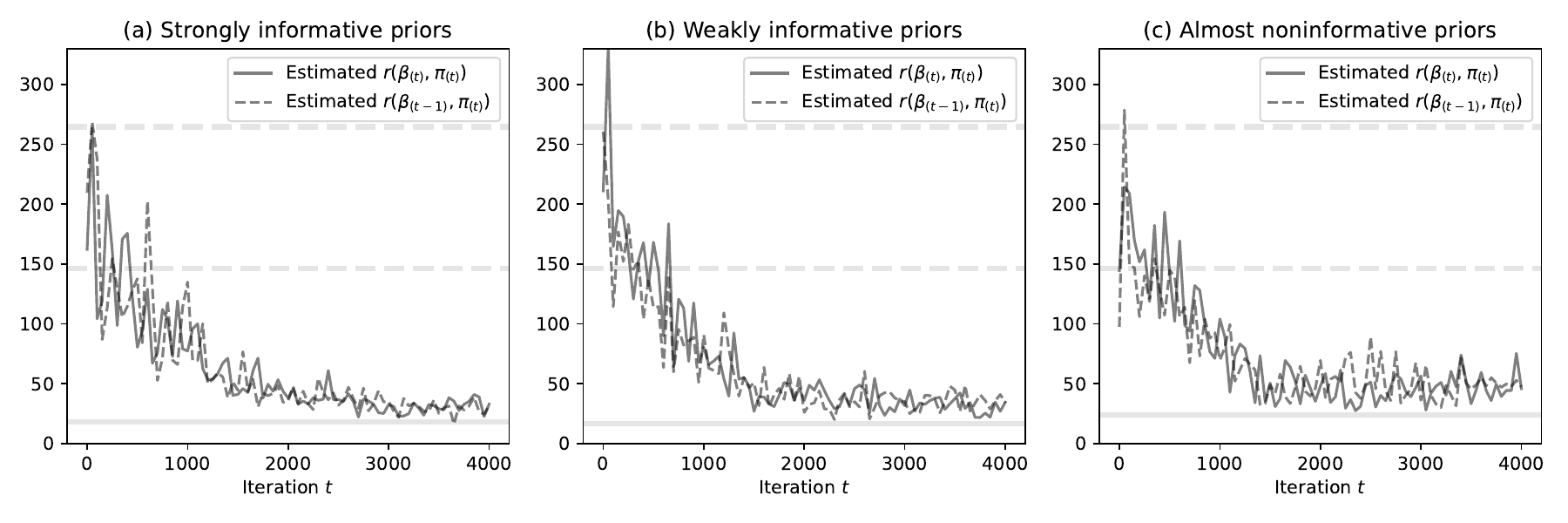}
	\caption{Estimated Bayes risks of the estimator over iterations when computing a $\Gamma_1$-minimax estimator. The lines are unbiased estimates of the current Bayes risks (y-axis) with 30 Monte Carlo runs over iterations (x-axis). The two dashed horizontal lines are the risks of the OSW (upper) and the SCL (lower) estimators, respectively, under $P_0$ in the simulation. The solid horizontal line is the risk of the computed $\Gamma$-minimax estimator under $P_0$. For clearness of visualization, the estimated Bayes risks are plotted every 50 iterations.}
	\label{figure: n new species Bayes risk and risk}
\end{figure}

We also apply the above methods to analyze the dataset studied in \cite{Miller1989}, which is used as the true population in the simulation. Based on this sample consisting of $n=1088$ observations in 188 categories, we use various methods to predict the number of new categories that would be observed if another $m=2000$ observations were to be collected. We train Gamma-minimax estimators using exactly the same tuning parameters as those in the above simulation, except that the starting point in Algorithm \ref{algorithm: MCMC-like} has more categories. The predictions of all methods are presented in Table~\ref{table: new species data result}. The $\Gamma$-minimax estimator outputs a 
more
similar prediction to the SCL estimator.
This similarity appears different from our observation in the simulation, but can be explained by the fact that having more observations ($n=1088$ vs $n=100$; $m=2000$ vs $m=200$) decreases the variance of the number of new observed categories and thus lowers discrepancies between predictions from these methods. 
Since the SCL estimator outperforms the OSW estimator in the above simulation where this dataset is the true population, we expect the SCL estimator to achieve reasonably good performance in this application. Moreover, given that the $\Gamma$-minimax estimators outperform the SCL estimator in the above simulation, we expect that 
57 or 58
represents an improved prediction of the number of new categories as compared to the SCL prediction of 51 in the case where there is limited prior information available.

\begin{table}[bt!]
	\centering
	\caption{Predicted number of new categories (rounded to the nearest integer) in a new sample with size 2000 based on the sample with size 1088 studied in \cite{Miller1989}. The strength of prior information in $\Gamma$-minimax estimators is shown in brackets.}
	\label{table: new species data result}
	\begin{tabular}{l|p{7em}}
		\hline
		Estimator & Predicted \#\phantom{jnk} new categories \\ \hline
		OSW & \hspace{2.5em}72 \\ %
		SCL & \hspace{2.5em}51 \\ %
		$\Gamma$-minimax (strong) & \hspace{2.5em}57 \\ %
		$\Gamma$-minimax (weak) & \hspace{2.5em}57 \\ %
		$\Gamma$-minimax (almost none) & \hspace{2.5em}58 \\ %
		\hline
	\end{tabular}
\end{table}

The computation time to compute an approximated $\Gamma$-minimax estimator was about five to seven hours on an AWS EC2 instance \citep{AmazonEC2} with at least 4 vCPUs and at least 8 GiB of memory, depending on the number of times the grid was enlarged. As shown in Fig.~\ref{figure: n new species Bayes risk and risk}, far few iterations are needed for SGDmax to output a good approximation of a $\Gamma_1$-minimax estimator, which is itself quite close to $\Gamma$-minimax. Therefore, with suitably less conservative tuning parameters or more adaptive minimax problem solvers, the computation time might drastically decrease. Moreover, the computation time needed to evaluate the computed $\Gamma$-minimax estimator at any sample is almost zero.

\subsection{Estimation of the entropy} \label{section: simulation entropy}

We also apply our method to estimate the entropy of a multinomial distribution (Example~\ref{example: estimate entropy}). The data-generating mechanism is the same as that described in Example~\ref{example: predict n new species}, and the estimand of interest is Shannon entropy \citep{Shannon1948}, that is, $\Psi(P_0) = -\sum_{k=1}^K p_k \log p_k$. In the simulation, we choose the same true population and the same sample size $n=100$ as in Section~\ref{section: simulation n new species}. The true entropy $\Psi(P_0)$ is $4.57$. As a reference, the entropy of the uniform distribution with the same number of categories---which corresponds to the maximum entropy of multinomial distributions with the same total number of categories---is $5.24$.

\citet{Jiao2015} developed a minimax rate optimal estimator of the Shannon entropy, and we run this simulation to investigate the potential gain of computing a Gamma-minimax estimator in well-posed problems with satisfactory solutions. As in Section~\ref{section: simulation n new species}, we consider three sets of prior information:
\begin{enumerate}
	\item Strongly informative: Prior mean of $\Psi(P)$ in $[4.3,4.7]$, $\geq 95\%$ probability that $\Psi(P)$ lies in $[4,5]$;
	\item Weakly informative: Prior mean of $\Psi(P)$ in $[4,5]$, $\geq 95\%$ probability that $\Psi(P)$ lies in $[3.5,5.5]$;
	\item Almost noninformative: Prior mean of $\Psi(P)$ in $[3.7,5.3]$, $\geq 95\%$ probability that $\Psi(P)$ lies in $[3,6]$.
\end{enumerate}

The architecture of our neural network estimator is almost identical to that in Section~\ref{section: simulation n new species} except that the existing estimator being used is the one proposed in \cite{Jiao2015} (referred to as the JVHW estimator), and we initialize the network to return the JVHW estimator. We use Algorithm~\ref{algorithm: MCMC-like} to construct $\modelspace_\ell$ and Algorithm~\ref{algorithm: convenient SGDmax} to compute a $\Gamma_\ell$-minimax estimator. The tuning parameters in the algorithms are identical to those used in Section~\ref{section: simulation n new species} except that, in Algorithm~\ref{algorithm: convenient SGDmax}, (i) the learning rate is $\eta=0.001$, and (ii) the number of iterations is 6,000 for $\Gamma_1$. We change these tuning parameters because the JVHW estimator is already minimax in terms of its convergence rate \citep{Jiao2015}, and we may need to update the estimator in a more cautious manner in Algorithm~\ref{algorithm: convenient SGDmax} to obtain any possible improvement.
The trajectories of the estimated Bayes risks (Fig.~\ref{figure: entropy Bayes risk and risk}) all appear to approximately reach a plateau, suggesting that the obtained estimator approximately $\Gamma_1$-minimax and that our choice of a smaller learning rate and a larger number of iterations is valid.
Because of the additional complexity of the JVHW estimator, we ran our simulations on an AWS EC2 instance \citep{AmazonEC2} with 4 vCPUs and 32 GiB of memory. The computation time was ten to seventeen hours, depending on the number of times the grid was enlarged. The longer computation time than that described in Section~\ref{section: simulation n new species} is primarily due to more iterations in SGDmax and the additional complexity of the JVHW estimator.

We compare the risk of the JVHW estimator and our trained $\Gamma$-minimax estimator under our set data-generating mechanism computed with 20000 Monte Carlo runs. We also compare their Bayes risk under the computed prior from Algorithm~\ref{algorithm: convenient SGDmax} using the last and finest grid in the computation with 20000 Monte Carlo runs. The results are summarized in Table~\ref{table: entropy result}. In this simulation experiment, our $\Gamma$-minimax estimator reduces the risk by a fair percentage compared with the JVHW estimator and achieves lower worst-case Bayes risk.
According to these simulation results, we conclude that incorporating informative prior knowledge into the estimator may result in some improvement in estimating entropy. Thus, for well-posed statistical problems with satisfactory solutions, we expect mild or no substantial improvement and little deterioration from using a Gamma-minimax estimator.

\begin{table}[bt!]
	\centering
	\caption{Risks and Bayes risks of estimators. $R(d,P_0)$: risk of the estimator under the true data-generating mechanism $P_0$. $r(d,\hat{\pi}^*)$: Bayes risk under prior $\hat{\pi}^*$, the computed prior from Algorithm~\ref{algorithm: convenient SGDmax} in the last and finest grid in the computation.}
	\label{table: entropy result}
	\begin{tabular}{l|l|r|r}
		\hline
		Strength of prior & Estimator & $R(d,P_0)$ & $r(d,\hat{\pi}^*)$ \\
		\hline
		strong & JVHW & 0.041 & 0.035 \\
		& $\Gamma$-minimax & 0.036 & 0.021 \\
		\rule{0pt}{4ex} weak & JVHW & 0.041 & 0.028 \\
		& $\Gamma$-minimax & 0.018 & 0.024 \\
		\rule{0pt}{4ex} almost none & JVHW & 0.041 & 0.031 \\
		& $\Gamma$-minimax & 0.025 & 0.016 \\
		\hline
	\end{tabular}
\end{table}

Fig.~\ref{figure: entropy Bayes risk and risk} presents the unbiased estimator of Bayes risks over iterations when computing a $\Gamma_1$-minimax estimator.
With strongly informative prior information present, the Bayes risks appear to fluctuate without an increasing or decreasing trend at the beginning and decrease slowly after several thousand iterations.
With weakly informative or almost no prior information, the Bayes risks also decrease slowly.
A reason may be that the JVHW estimator is already minimax rate optimal \citep{Jiao2015}.
The computed $\Gamma$-minimax estimators also appear to be somewhat similar to the JVHW estimator: in the output layer of the three settings with different prior information, the coefficients for the JVHW estimator are $0.97$, $0.90$ and $0.89$, respectively; the coefficients for the previous hidden layer are $0.17$, $0.17$ and $0.20$, respectively; the intercepts are $0.06$, $0.30$ and $0.30$, respectively.

\begin{figure}[bt!]
	\centering
	\includegraphics[scale=0.4]{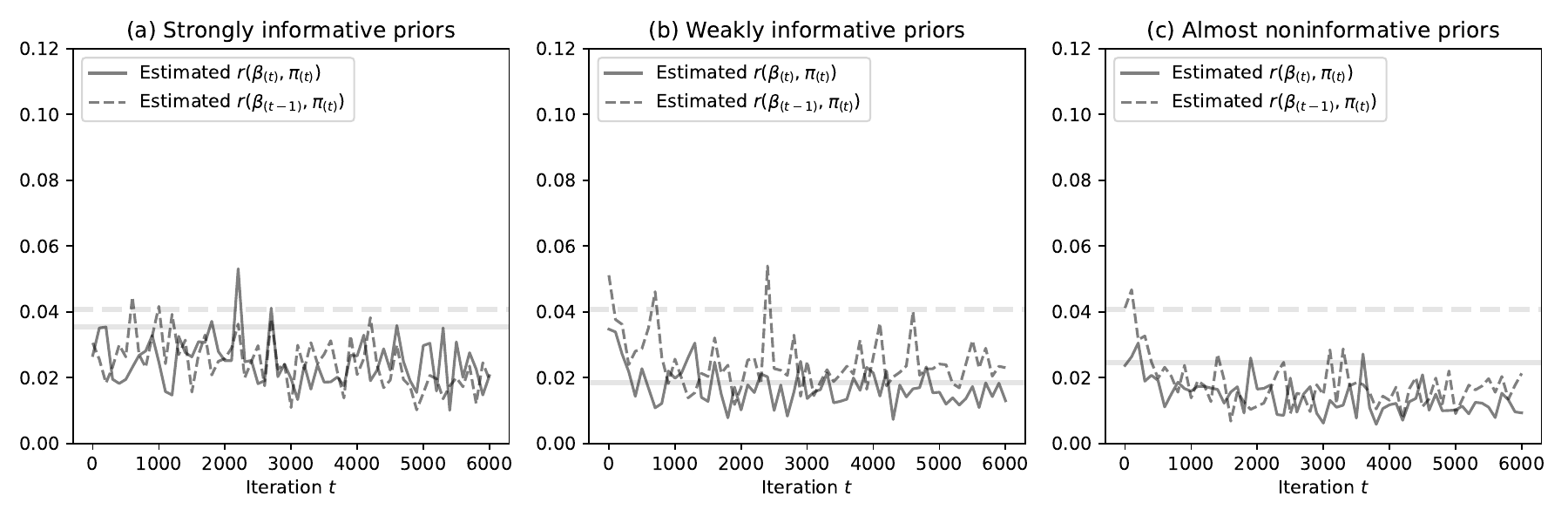}
	\caption{Estimated Bayes risks of the estimator over iterations when computing a $\Gamma_1$-minimax estimator. The lines are unbiased estimates of the current Bayes risks (y-axis) with 30 Monte Carlo runs over iterations (x-axis). The horizontal lines are the risks of the JVHW (dashed) and the computed $\Gamma$-minimax (solid) estimators, respectively, under $P_0$ in the simulation. For clearness of visualization, the estimated Bayes risks are plotted every 100 iterations.}
	\label{figure: entropy Bayes risk and risk}
\end{figure}

We further use the above methods to estimate entropy based on the dataset used as the true population in the simulation. The tuning parameters of the $\Gamma$-minimax estimators are exactly the same as those in the above simulation except that the starting point in Algorithm \ref{algorithm: MCMC-like} has more categories. The estimates are presented in Table~\ref{table: entropy data result}. All methods produce almost identical estimates. Because the sample size is more than ten times the sample size in the simulation and the JVHW estimator is minimax rate optimal \citep{Jiao2015}, we expect the JVHW estimator to have little room for improvement, which explains why the three $\Gamma$-minimax estimators perform similarly to the JVHW estimator. In other words, Gamma-minimax estimators appear to maintain, if not improve, the performance of the original JVHW estimator.

\begin{table}[bt!]
	\centering
	\caption{Estimated entropy based on the sample with size 1088 studied in \cite{Miller1989}. The strength of prior information in $\Gamma$-minimax estimators is shown in brackets.}
	\label{table: entropy data result}
	\begin{tabular}{l|l}
		\hline
		Estimator & Estimated entropy \\ \hline
		JVHW & \hspace{2.5em} 4.709 \\
		$\Gamma$-minimax (strong) & \hspace{2.5em} 4.709 \\
		$\Gamma$-minimax (weak) & \hspace{2.5em} 4.708 \\
		$\Gamma$-minimax (almost none) & \hspace{2.5em} 4.703 \\
		\hline
	\end{tabular}
\end{table}

\section{Discussion} \label{section: discussion}

We propose adversarial meta-learning algorithms to compute a Gamma-minimax estimator with theoretical guarantees under fairly general settings. These algorithms still leave room for improvement. As we discussed in Section~\ref{section: increasing grid}, the stopping criterion we employ does not necessarily indicate that the maximal Bayes risk is close to the true minimax Bayes risk. In future work, it would be interesting to derive a better criterion that necessarily does indicate this near optimality. Our algorithms also require the user to choose increasingly fine approximating grids to the model space. Although we propose a heuristic algorithm for this procedure that performed well in our experiments, at this point, we have not provided optimality guarantees for this scheme. It may also be possible to improve our proposed algorithms to solve intermediate minimax problems in Section~\ref{section: increasing grid} by utilizing recent and ongoing advances from the machine learning literature that can be used to improve the training of generative adversarial networks.

We do not explicitly consider uncertainty quantification such as confidence intervals or credible intervals under a Gamma-minimax framework. Uncertainty quantification is important in practice since it provides more information than a point estimator and can be used for decision-making. In theory, our method may be directly applied if such a problem can be formulated into a Gamma-minimax problem. However, such a formulation remains unclear. The most challenging part is to identify a suitable risk function that correctly balances the level of uncertainty and the size of the output interval/region. Though the risk function used in \cite{Schafer2009} appears to provide one possible starting point, it is not clear how to extend this approach to nonparametric settings.

It is possible to allow the space of estimators $\decisionspace$ to increase as the grid $\modelspace_\ell$ increase. For example, we may specify an increasing sequence of estimator spaces $\{\decisionspace_\ell\}_{\ell=1}^\infty$ whose limit is dense in a general space $\decisionspace_0$; then, in Line~3 of Algorithm~\ref{algorithm: approximation with increasingly fine grid}, we compute a $\Gamma_\ell$-minimax estimator in $\decisionspace_\ell$, namely replace $\decisionspace$ with $\decisionspace_\ell$. This sequence of estimators might converge to a $\Gamma$-minimax estimator in $\decisionspace_0$. One possible choice of $\decisionspace_\ell$ ($\ell>1$) in this approach is a space of statistical knowledge networks with the given estimator being the computed $\Gamma_{\ell-1}$-minimax estimator in $\decisionspace_{\ell-1}$. It is of future interest to investigate the properties of such an approach.

In conclusion, we propose adversarial meta-learning algorithms to compute a Gamma-minimax estimator under general models that can incorporate prior information in the form of generalized moment conditions. They can be useful when a parametric model is undesirable, semi-parametric efficiency theory does not apply, or we wish to utilize prior information to improve estimation.

\section*{Acknowledgements}
Generous support was provided by Amazon through an AWS Machine Learning Research Award, the NIH under award number DP2-LM013340, and the NSF under award number DMS-2210216. The content is solely the responsibility of the authors and does not necessarily represent the official views of Amazon, the NIH, or the NSF.

\newpage

\begin{appendix}
	
	\section{Two counterexamples of Condition~\ref{condition: subset of modelspace}} \label{section: counterexample of subset of modelspace}
	
	We provide two counterexamples of Condition~\ref{condition: subset of modelspace} to illustrate that this condition fails in extremely ill cases.
	
	In the first counterexample, $P \mapsto R(d,P)$ is discontinuous: we set $R(d,P^*)$ to be zero for a fixed $P^* \in \modelspace$ and $R(d,P)$ to be one for all other $P \in \modelspace$. If we choose the grid $\modelspace_\ell$ to be dense in $\modelspace$ but to never contain $P^*$, then Condition~\ref{condition: subset of modelspace} does not hold since $r_{\sup}(d,\tilde{\Gamma}_\ell)=1$ for sufficiently large $\ell$ such that $P^* \in \Omega_\ell$ but $r_{\sup}(d,\tilde{\Gamma}_{i \mid \ell})=0$ for all $i$ and $\ell$. This issue can be resolved by choosing a continuous risk function.
	
	In the second counterexample, $\modelspace_\ell$ does not contain distributions that are consistent with prior information. Suppose that $\Gamma=\{\pi \in \Pi: \int \Phi(P) \pi(\intd P) = 0\}$ where $\Phi(P):=\expect_P[X^2]$. In other words, it is known that the true data-generating mechanism $P_0$ must be a distribution that is a point mass at zero, and thus $\Gamma$ also only contains a point mass at $P_0$. If $\Phi(P) \neq 0$ for every $P \in \cup_{i=1}^\infty \modelspace_i$, then, even if $\bigcup_{\ell=1}^\infty \modelspace_\ell$ is dense in $\modelspace$, $\tilde{\Gamma}_{i \mid \ell} = \emptyset$ and thus Condition~\ref{condition: subset of modelspace} does not hold. This issue can be resolved by rewriting the problem such that these hard constraints on $\modelspace$ are incorporated into the specification of $\modelspace$ rather than $\Gamma$.

	\section{Additional gradient-based algorithms} \label{section: GDmax}
	
	If we can evaluate $R(\beta,P)$ exactly for all $\beta \in \mathcal{H}$ and $P \in \modelspace_\ell$, then the GDmax algorithm (Algorithm~\ref{algorithm: GDmax}) may be used. Note that Line~\ref{step: GDmax maximization} can be formulated into a linear program, which can always be solved in polynomial time with an interior point method \citep[e.g.,][]{Jiang2020} and often be solved in polynomial time with a simplex method \citep{Spielman2004}.
	
	\begin{algorithm}
		\caption{Gradient descent with max-oracle (GDmax) to compute a $\Gamma_\ell$-minimax estimator}
		\label{algorithm: GDmax}
		\begin{algorithmic}[1]
			\State Initialize $\beta_{(0)} \in \real^D$. Set learning rate $\eta>0$ and max-oracle accuracy $\zeta>0$.
			\For{$t=1,2,\ldots$}
			\State Maximization: find $\pi_{(t)} \in \Gamma_\ell$ such that $r(\beta_{(t-1)},\pi_{(t)}) \geq \max_{\pi \in \Gamma_\ell} r(\beta_{(t-1)},\pi) - \zeta$ \label{step: GDmax maximization}
			\State Gradient descent: $\beta_{(t)} \gets \beta_{(t-1)} - \eta \nabla_\beta r(\beta,\pi_{(t)}) |_{\beta = \beta_{(t-1)}}$ \label{step: GDmax GD}
			\EndFor
		\end{algorithmic}
	\end{algorithm}
	
	We have the following result on the validity of GDmax.
	\begin{theorem}[Validity of GDmax (Algorithm~\ref{algorithm: GDmax})] \label{theorem: GDmax convergence}
		Under conditions in Theorem~\ref{theorem: SGDmax convergence}, in Algorithm~\ref{algorithm: GDmax}, with $\eta = \epsilon^2/(L_1 L_2^2)$ and $\zeta = \epsilon^2/(24 L_1)$, $\beta_{(t)}$ is an $\epsilon$-stationary point of $\beta \mapsto r_{\sup}(\beta,\Gamma_\ell)$ for $t = O(L_1 L_2 \Delta/\epsilon^4)$, and is thus close to a local minimum of $\beta \mapsto r_{\sup}(\beta,\Gamma_\ell)$.
	\end{theorem}

	Therefore, we propose a variant (Algorithm~\ref{algorithm: convenient SGDmax}) by replacing this line with Lines~\ref{step: convenient SGDmax stochastic max1}--\ref{step: convenient SGDmax stochastic max2} so that ordinary linear program solvers can be directly applied. The following theorem justifies this variant.
	
	\begin{algorithm}
		\caption{Convenient variant of SGDmax (Algorithm~\ref{algorithm: SGDmax}) to compute a $\Gamma_\ell$-minimax estimator}
		\label{algorithm: convenient SGDmax}
		\begin{algorithmic}[1]
			\State Initialize $\beta_{(0)} \in \real^D$. Set learning rate $\eta>0$ and batch sizes $J$, $J'$.
			\For{$t=1,2,\ldots$}
			\State Generate iid copies $\xi_1,\ldots,\xi_{J'}$ of $\xi$. \label{step: convenient SGDmax stochastic max1}
			\State Stochastic maximization: $\pi_{(t)} \gets \argmax_{\pi \in \Gamma_\ell} \frac{1}{J'} \sum_{j=1}^{J'} \hat{r}(\beta_{(t-1)},\pi,\xi_j)$. \label{step: convenient SGDmax stochastic max2}
			\State Generate iid copies of $\xi_{J'+1},\ldots,\xi_{J'+J}$ of $\xi$.
			\State Stochastic gradient descent: $\beta_{(t)} \gets \beta_{(t-1)} - \frac{\eta}{J} \sum_{j=J'+1}^{J'+J} \nabla_\beta \hat{r}(\beta,\pi_{(t)},\xi_j) |_{\beta = \beta_{(t-1)}}$.
			\EndFor
		\end{algorithmic}
	\end{algorithm}
	
	The validity of this variant of SGDmax is given in Theorem~\ref{theorem: SGDmax variant convergence} below.
	
	\begin{theorem}[Validity of convenient variant of SGDmax (Algorithm~\ref{algorithm: convenient SGDmax})] \label{theorem: SGDmax variant convergence}
		Suppose that $\{\xi \mapsto \hat{r}(\beta,\pi,\xi): \beta \in \real^D, \pi \in \Gamma_\ell\}$ is a $\Xi$-Glivenko-Centelli class \citep{vandervaart1996}. Then, for any $\zeta>0$, there exists a sufficiently large $J'$ such that
		$$\expect[r(\beta_{(t-1)},\pi_{(t)})] \geq \max_{\pi \in \Gamma_\ell} r(\beta_{(t-1)},\pi) - \zeta$$
		for all $t$, where the expectation is taken over $\pi_{(t)}$ and $\beta_{(t-1)}$ is fixed. Therefore, with the chosen parameters in Theorem~\ref{theorem: SGDmax convergence}, we may choose a sufficiently large $J'$ so that $\beta_{(t)}$ is an $\epsilon$-stationary point of $\beta \mapsto r_{\sup}(\beta,\Gamma_\ell)$ in expectation for $t = O(L_1 (L_2^2 + \sigma^2) \Delta/\epsilon^4)$ and is thus close to a local minimum of $\beta \mapsto r_{\sup}(\beta,\Gamma_\ell)$ with high probability.
	\end{theorem}
	
	We prove Theorem~\ref{theorem: SGDmax variant convergence} by showing that $\max_{\pi \in \Gamma_\ell} r(\beta_{(t-1)},\pi) - \expect[r(\beta_{(t-1)},\pi_{(t)})]$ converges to $0$ as $J' \rightarrow \infty$. The proof is essentially an application of empirical process theory to the study of an M-estimator.

	\section{Additional simulation: mean estimation} \label{section: simulation mean}
	
	In this appendix, we illustrate our proposed methods via simulation in a special case of Example~\ref{example: estimation}, namely for estimating the mean of a distribution. We assume that $\modelspace$ consists of all probability distributions defined on the Borel $\sigma$-algebra on $[0,1]$ and we observe $\observation=(X_1,X_2,\ldots,X_n)$, where $X_1,\ldots,X_n \overset{\mathrm{iid}}{\sim} P_0 \in \modelspace$. Here we take $n=10$. The estimand is $\Psi(P_0)=\int x \, P_0(\intd x)$. We use the mean squared error risk introduced in Example~\ref{example: estimation}. Suppose that we represent the prior information by $\Gamma=\{\pi \in \Pi: \int \Psi(P) \, \pi(\intd P) = 0.3 \}$, which corresponds to the set of prior distributions in $\Pi$ that satisfy an equality constraint on the prior mean of $\Psi(P)$.
	
	We apply our method to three spaces of estimators separately. The first space, $\decisionspace_\mathrm{linear}$, is the set of affine transformations of the sample mean, that is, $\decisionspace_\mathrm{linear}=\{d: d(\observation) = \beta_0 + \beta_1 \sum_{i=1}^n X_i / n, \beta_0,\beta_1 \in \real\}$. As shown in Proposition~\ref{proposition: Gamma-minimax estimator of mean} in Appendix~\ref{section: theory of mean estimation}, there is an estimator $d^*$ in $\decisionspace_\mathrm{linear}$ that is $\Gamma$-minimax in the space of all estimators that are square-integrable with respect to all $P \in \modelspace$, so we consider this simple space to better compare our computed estimator with that theoretical $\Gamma$-minimax estimator. When computing a $\Gamma$-minimax estimator in $\decisionspace_\mathrm{linear}$, we initialize the estimator to be the sample mean, that is, we let $\beta_0=0$ and $\beta_1=1$.
	
	The second space, $\decisionspace_\mathrm{skn}$ (statistical knowledge network), is a set of neural networks designed based on statistical knowledge that includes the sample mean as an input. We consider this space to illustrate our proposal in Section~\ref{section: choice of estimator space}. More precisely, we use the architecture in Fig.~\ref{figure: skn mean}, which is similar to the deep set architecture \citep{Zaheer2017,Maron2019} and is a permutation invariant neural network. We use such an architecture to account for the fact that the sample is iid. In this architecture, the sample mean node is used as an augmenting node to an ordinary deep set network and is combined with the output of that ordinary network in the fourth hidden layer to obtain the final output. Note that $\decisionspace_\mathrm{skn}\supset\decisionspace_\mathrm{linear}$. When computing a $\Gamma$-minimax estimator for this class, we also initialize the network to be exactly the sample mean, which is a reasonable choice given that the sample mean is known to be a sensible estimator. In this simulation experiment, we choose the dimensionality of nodes in each hidden layer in Fig.~\ref{figure: skn mean} as follows: each node in the first, second, third and fourth hidden layer represents a vector in $\real^{10}$, $\real^5$, $\real^{10}$ and $\real$, respectively. We do not use larger architectures because usually the sample mean is already a good estimator, and we expect to obtain a useful estimator as a small perturbation of this estimator. We also use the ReLU as the activation function. We did not use ELMs in this and the following simulations because we found that neural networks perform well.
	
	\begin{figure}[bt!]
		\centering
		\resizebox{!}{2.3in}{%
			\begin{tikzpicture}[x=1.5cm, y=1.5cm, >=stealth]
				\foreach \m [count=\y] in {1,2,3,missing,4}
				\node [every neuron/.try, neuron \m/.try] (input-\m) at (0,2.5-\y) {};
				
				\foreach \m [count=\y] in {1,2,3,missing,4}
				\node [every neuron/.try, neuron \m/.try ] (transformhidden1-\m) at (1.5,2.5-\y) {};
				
				\foreach \m [count=\y] in {1,2,3,missing,4}
				\node [every neuron/.try, neuron \m/.try ] (transformhidden2-\m) at (3,2.5-\y) {};
				
				
				\node [every neuron/.try, neuron 1/.try ] (aggregate-1) at (4.5,1) {};
				\node [every neuron/.try, neuron 2/.try ] (aggregate-2) at (4.5,-1) {$\sum$};
				
				\foreach \m [count=\y] in {1}
				\node [every neuron/.try, neuron \m/.try ] (aggregatehidden1-\m) at (6,-1) {};
				
				\foreach \m [count=\y] in {1}
				\node [every neuron/.try, neuron \m/.try ] (aggregatehidden2-\m) at (7.5,-1) {};
				
				\foreach \m [count=\y] in {1}
				\node [every neuron/.try, neuron \m/.try ] (output-\m) at (9,0) {};
				
				\foreach \l [count=\i] in {1,2,3,n}
				\draw [<-] (input-\i) -- ++(-1,0)
				node [above, midway] {$X_\l$};
				
				\node [above] at (aggregate-1.north) {Sample mean};
				
				\foreach \i in {1,...,4}
				\draw [->] (input-\i) -- (transformhidden1-\i);
				
				\foreach \i in {1,...,4}
				\draw [->] (transformhidden1-\i) -- (transformhidden2-\i);
				
				\foreach \i in {1,...,4}
				\draw [->] (transformhidden2-\i) -- (aggregate-2);
				
				\draw [->] (aggregate-1) -- (aggregatehidden1-1);
				\draw [->] (aggregate-2) -- (aggregatehidden1-1);
				
				\draw [->] (aggregate-1) -- (aggregatehidden2-1);
				\draw [->] (aggregatehidden1-1) -- (aggregatehidden2-1);
				
				\draw [->] (aggregate-1) -- (output-1);
				\draw [->] (aggregatehidden2-1) -- (output-1);
				
				\foreach \l [count=\x from 0] in {Input, 1st Hidden, 2nd Hidden, Pooling, 3rd Hidden, 4th Hidden, Ouput}
				\node [align=center, above] at (\x*1.5,2) {\l \\ layer};
			\end{tikzpicture}%
		}
		\caption{Architecture of the permutation invariant neural network estimator of the mean in $\decisionspace_\mathrm{skn}$. $X_i$: observation $i$ in the sample; $\sum$: the node that sums up all ancestor nodes. In the first two hidden layers, all input nodes are transformed by the same function. The arrows from the input nodes to the sample mean estimator are omitted from this graph. Each node in the hidden layers represents a vector.}
		\label{figure: skn mean}
	\end{figure}
	
	The third space, $\decisionspace_\mathrm{nn}$, is a set of neural networks that do not utilize knowledge of the sample mean. We consider this space to illustrate our method without utilizing existing estimators. These estimators are also deep set networks with similar architecture as $\decisionspace_\mathrm{skn}$ in Fig.~\ref{figure: skn mean}. The main difference is that the explicit sample mean node and the fourth hidden layer are removed. When computing a $\Gamma$-minimax estimator in $\decisionspace_\mathrm{nn}$, we also randomly initialize the network, unlike $\decisionspace_\mathrm{linear}$ and $\decisionspace_\mathrm{skn}$, in order not to input statistical knowledge. Because the ReLU activation function is used, $\decisionspace_\mathrm{nn}\supset \decisionspace_\mathrm{linear}$, and we do not expect that optimizing over $\decisionspace_\mathrm{nn}$ should not lead to a $\Gamma$-minimax estimator with worse performance than those in $\decisionspace_\mathrm{linear}$ and $\decisionspace_\mathrm{skn}$.
	
	To construct the grid $\modelspace_\ell$ for this problem, we use a simpler method than Algorithm~\ref{algorithm: MCMC-like}. As indicated by Lemma~\ref{lemma: max variance distribution} in Appendix~\ref{section: theory of mean estimation}, for estimators in $\decisionspace_\mathrm{linear}$, Bernoulli distributions tend to have high risks since all probability weights lie on the boundary of $[0,1]$; in addition, a prior $\pi^*$ for which $d^*$ is Bayes is a Beta prior over Bernoulli distributions. Therefore, we randomly generate 2000 Bernoulli distributions as grid points in $\modelspace_1$. We also include two degenerate distributions in this grid, namely the distribution that places all of its mass at $0$ and that which places all of its mass at $1$. When constructing $\modelspace_\ell$ from $\modelspace_{\ell-1}$, we still add in more complicated distributions to make the grid dense in the limit: we first randomly generate 500 discrete distributions with support being those in $\modelspace_{\ell-1}$; then we randomly generate 10 new support points in $[0,1]$ and 1000 distributions with support points being the union of the new support points and the existing support points in $\modelspace_{\ell-1}$.
	
	When computing the $\Gamma$-minimax estimator, for each grid $\modelspace_\ell$, we compute the $\Gamma_\ell$-minimax estimator for all three estimator spaces with Algorithm~\ref{algorithm: convenient SGDmax}. We set the learning rate $\eta = 0.005$, the batch size $J=50$ and the number of iterations to be 200 for $\Gamma_\ell$ ($\ell>1$). The number of iterations for $\Gamma_1$ is larger because, in our experiments, we saw that a $\Gamma_1$-minimax estimator is already close to a $\Gamma$-minimax estimator, and using a large number of iterations in this step can improve the initial estimator substantially. For $\decisionspace_\mathrm{linear}$ and $\decisionspace_\mathrm{skn}$, the number of iterations for $\Gamma_1$ is 2000; the corresponding number for $\decisionspace_\mathrm{nn}$ is 6000 to account for the lack of human knowledge input. We also use Algorithm~\ref{algorithm: fictitious play} with 10000 iterations to compute a $\Gamma_\ell$-minimax estimator for $\decisionspace_\mathrm{linear}$ for illustration. In this setup, as described in Section~\ref{section: fictitious play}, we take the average of the computed $\Gamma$-minimax stochastic estimator as the final output estimator in $\decisionspace_\mathrm{linear}$. We do not apply Algorithm~\ref{algorithm: fictitious play} to $\decisionspace_\mathrm{skn}$ or $\decisionspace_\mathrm{nn}$ because it is computationally intractable for these estimator spaces.
	
	We set the stopping criterion in Algorithm~\ref{algorithm: approximation with increasingly fine grid} as follows. When Algorithm~\ref{algorithm: convenient SGDmax} is used to compute $\Gamma_\ell$-minimax estimators, we estimate both $r_{\sup}(d^*_{\ell-1},\Gamma_{\ell})$ and $r_{\sup}(d^*_{\ell-1},\Gamma_{\ell-1})$ with 2000 Monte Carlo runs as described in Section~\ref{section: increasing grid}; when Algorithm~\ref{algorithm: fictitious play} is used, $r_{\sup}(d^*_{\ell-1},\Gamma_{\ell})$ and $r_{\sup}(d^*_{\ell-1},\Gamma_{\ell-1})$ are computed exactly because $R(d,P)$ has a closed-form expression for all $d \in \decisionspace_\mathrm{linear}$ and $P \in \modelspace_\ell$. We set the tolerance $\epsilon$ to be equal to $0.0001$ so that we stop Algorithm~\ref{algorithm: approximation with increasingly fine grid} if $r_{\sup}(d^*_{\ell-1},\Gamma_{\ell}) - r_{\sup}(d^*_{\ell-1},\Gamma_{\ell-1}) \leq \epsilon$.
	
	After computation, we report the Bayes risk of the computed and theoretical $\Gamma$-minimax estimators under $\pi^*$, the prior such that $r(d^*,\pi^*) = \inf_{d \in \decisionspace} r_{\sup}(d,\Gamma)$. For the estimators in $\decisionspace_\mathrm{linear}$, we further report their coefficients. We also report two coefficients of the computed estimator in $\decisionspace_\mathrm{skn}$ as follows. Since $\decisionspace_\mathrm{linear} \subseteq \decisionspace_\mathrm{skn}$ and we initialize the estimator to be the sample mean for $\decisionspace_\mathrm{skn}$, we would expect that the bias $\beta_0$ and the weight of the sample mean $\beta_1$ in the output layer for the computed $\Gamma$-minimax estimator in $\decisionspace_\mathrm{skn}$ may correspond to those in $\decisionspace_\mathrm{linear}$. Therefore, we also report these two coefficients $\beta_0$ and $\beta_1$ for $\decisionspace_\mathrm{skn}$. This may not be the case for $\decisionspace_\mathrm{nn}$ because the sample mean is not explicit in its parameterization and all coefficients are randomly initialized, so we do not report any coefficients for $\decisionspace_\mathrm{nn}$.
	
	Table~\ref{table: mean result} presents the computation results. By Theorem~\ref{theorem: Gamma-minimaxity criterion} in Appendix~\ref{section: theory of mean estimation}, these computed estimators are all approximately $\Gamma$-minimax since their Bayes risks for $\pi^*$ are all close to that of a theoretical $\Gamma$-minimax estimator. The coefficients $\beta_0$ and $\beta_1$ of the computed estimators in $\decisionspace_\mathrm{linear}$ and $\decisionspace_\mathrm{skn}$ are also close to a theoretically derived estimator. For the computed estimator in $\decisionspace_\mathrm{skn}$, the weight of the other ancestor node in the output layer (i.e., the node in the 4th hidden layer in Fig.~\ref{figure: skn mean}) is $0.000$. Therefore, our computed $\Gamma$-minimax estimator in $\decisionspace_\mathrm{skn}$ is also close to a theoretically derived $\Gamma$-minimax estimator.
	
	\begin{table}[bt!]
		\centering
		\caption{Coefficients and Bayes risks of estimators of the mean. Unrestricted space: the space of all estimators that are square-integrable with respect to all $P \in \modelspace$.}
		\label{table: mean result}
		\begin{tabular}{l|l|r|r|r}
			\hline
			Estimator space & Method to obtain $d^*$ & $\beta_0$ & $\beta_1$ & $r(d,\pi^*)$ \\
			\hline
			Unrestricted space & Theoretical derivation & $0.072$ & $0.760$ & $0.012$ \\
			$\decisionspace_\mathrm{linear}$ & Algorithms~\ref{algorithm: approximation with increasingly fine grid} \& \ref{algorithm: convenient SGDmax} & $0.072$ & $0.763$ & $0.012$ \\
			$\decisionspace_\mathrm{skn}$ & Algorithms~\ref{algorithm: approximation with increasingly fine grid} \& \ref{algorithm: convenient SGDmax} & $0.071$ & $0.767$ & $0.012$ \\
			$\decisionspace_\mathrm{nn}$ & Algorithms~\ref{algorithm: approximation with increasingly fine grid} \& \ref{algorithm: convenient SGDmax} & --- & --- & $0.012$ \\
			$\decisionspace_\mathrm{linear}$ & Algorithms~\ref{algorithm: approximation with increasingly fine grid} \& \ref{algorithm: fictitious play} & $0.072$ & $0.760$ & $0.012$ \\
			\hline
		\end{tabular}
	\end{table}
	
	In our experiments, Algorithm~\ref{algorithm: approximation with increasingly fine grid} converged after computing a $\Gamma_1$-minimax estimator except when using Algorithm~\ref{algorithm: convenient SGDmax} for $\decisionspace_\mathrm{linear}$. Even in this exceptional case, the computed $\Gamma_1$-minimax estimator is still approximately $\Gamma$-minimax. We think the algorithm does not stop then in these cases because of Monte Carlo errors when computing $r_{\sup}(d^*_{\ell-1},\Gamma_\ell)$ and $r_{\sup}(d^*_{\ell-1},\Gamma_{\ell-1})$.
	
	Fig.~\ref{figure: mean parameter Bayes risk for least favorable prior} presents the Bayes risks (or its unbiased estimates) over iterations when computing a $\Gamma_1$-minimax estimator. In all cases using Algorithm~\ref{algorithm: convenient SGDmax}, the Bayes risks appear to decrease and converge. When using Algorithm~\ref{algorithm: fictitious play}, the upper and lower bounds both converge to the same limit. The limiting values of the Bayes risks in all cases are close to $r(d^*,\pi^*)$ because $\Gamma_1$ can approximate $\pi^*$ well.

		\section{Sensitivity analysis for tuning parameter selection} \label{section: simulation n new species sensitivity}
		
		For the simulation in Section~\ref{section: simulation n new species} with strongly informative prior information, we conduct three simulations to investigate the sensitivity of our proposed method to the selection of tuning parameters. In each simulation below, we vary one set of tuning parameters and rerun the algorithm to obtain an estimator. In the first simulation, we vary the starting point of Algorithm~\ref{algorithm: MCMC-like} to construct the first grid $\modelspace_1$. The new starting point is a distribution with 173 categories and $\Phi(P_{(0)}) = 61$, and so this starting point is qualitatively different from the one chosen in the original simulation. In the second simulation, we vary the grid sizes: There are 500 grid points in $\modelspace_1$ and we add 500 grid points each time we enlarge the grid. In the third simulation, we chose a wider and deeper statistical knowledge network (see Fig.~\ref{figure: nn n new species2}): Compared to the original simulation, we add one more hidden layer and increased the number of hidden nodes in the first two hidden layers to 100. As shown in Table~\ref{table: new species sensitivity result}, the results in these sensitivity simulations appear similar to that in Section~\ref{section: simulation n new species} within the variation due to randomness in MCMC (Algorithm~\ref{algorithm: MCMC-like}) and SGDmax (Algorithm~\ref{algorithm: convenient SGDmax}).

		\begin{figure}[bt!]
			\centering
			\resizebox{!}{2.3in}{%
				\begin{tikzpicture}[x=1.5cm, y=1.5cm, >=stealth]
					\foreach \m [count=\y] in {1,missing,2}
					\node [every neuron/.try, neuron \m/.try] (input-\m) at (0,0.5-\y) {};
					
					\node [every neuron/.try, neuron 1/.try] (OSW) at (0,2) {};
					\node [every neuron/.try, neuron 1/.try] (SCL) at (0,1) {};
					
					\foreach \m [count=\y] in {1,missing,2}
					\node [every neuron/.try, neuron \m/.try ] (hidden1-\m) at (2,0.5-\y) {};
					
					\foreach \m [count=\y] in {1}
					\node [every neuron/.try, neuron \m/.try ] (hidden2-\m) at (4,-1) {};
					
					\foreach \m [count=\y] in {1}
					\node [every neuron/.try, neuron \m/.try ] (hidden3-\m) at (6,-0.5) {};
					
					\foreach \m [count=\y] in {1}
					\node [every neuron/.try, neuron \m/.try ] (output-\m) at (8,0) {};
					
					\foreach \l [count=\i] in {1,n}
					\draw [<-] (input-\i) -- ++(-1,0)
					node [above, midway] {$X_\l$};
					
					\node [left] at (OSW.west) {OSW};
					\node [left] at (SCL.west) {SCL};

					\foreach \i in {1,...,2}
					\foreach \j in {1,...,2}
					\draw [->] (input-\i) -- (hidden1-\j);
					
					\foreach \j in {1,...,2}
					\draw [->] (OSW) -- (hidden1-\j);
					
					\foreach \j in {1,...,2}
					\draw [->] (SCL) -- (hidden1-\j);
					
					\foreach \i in {1,...,2}
					\foreach \j in {1}
					\draw [->] (hidden1-\i) -- (hidden2-\j);
					
					\foreach \i in {1}
					\foreach \j in {1}
					\draw [->] (hidden2-\i) -- (hidden3-\j);
					
					\draw [->] (OSW) -- (hidden2-1);
					\draw [->] (SCL) -- (hidden2-1);
					
					\draw [->] (OSW) -- (hidden3-1);
					\draw [->] (SCL) -- (hidden3-1);
					
					\draw [->] (OSW) -- (output-1);
					\draw [->] (SCL) -- (output-1);
					\draw [->] (hidden3-1) -- (output-1);
					
					\foreach \l [count=\x from 0] in {Input, 1st Hidden, 2nd Hidden, 3rd Hidden, Ouput}
					\node [align=center, above] at (\x*2,2.5) {\l \\ layer};
				\end{tikzpicture}%
			}
			\caption{Architecture of the deeper and wider neural network estimator of the expected number of new categories.}
			\label{figure: nn n new species2}
		\end{figure}

		\begin{table}[bt!]
			\centering
			\caption{Table similar to Table~\ref{table: new species result} for sensitivity analysis with strongly informative prior information.}
			\label{table: new species sensitivity result}
			\begin{tabular}{l|r|r}
				\hline
				Varied tuning parameter & $R(d,P_0)$ & $r(d,\hat{\pi}^*)$ \\
				\hline
				Initial distribution in MCMC & 19 & 44 \\
				Grid size & 15 & 34 \\
				Statistical knowledge network structure & 17 & 38 \\
				\hline
			\end{tabular}
		\end{table}

	\section{Proofs} \label{section: proof}
	
	\subsection{Proof of Theorem~\ref{theorem: approximate Gamma-minimax on a grid} and Corollary~\ref{corollary: convexity; approximate Gamma-minimax on a grid}} \label{section: Proof of Theorem approximate Gamma-minimax on a grid and Corollary convexity; approximate Gamma-minimax on a grid}
	
	\begin{lemma} \label{lemma: approximate r_sup on subset}
		If $\{\Omega_\ell\}_{\ell=1}^\infty$ is an increasing sequence of subsets of $\modelspace$ such that $\bigcup_{\ell=1}^\infty \Omega_\ell = \modelspace$, then, for any $d \in \decisionspace$, $r_{\sup}(d,\tilde{\Gamma}_\ell) \nearrow r_{\sup}(d,\Gamma)$ ($\ell \rightarrow \infty$).
	\end{lemma}
	
	\begin{proof}[Proof of Lemma~\ref{lemma: approximate r_sup on subset}]
		Since $\tilde{\Gamma}_\ell \subseteq \tilde{\Gamma}_{\ell+1} \subseteq \Gamma$, it holds that $r_{\sup}(d,\tilde{\Gamma}_\ell) \leq r_{\sup}(d,\tilde{\Gamma}_{\ell+1}) \leq r_{\sup}(d,\Gamma)$, and so we only need to lower bound $r_{\sup}(d,\tilde{\Gamma}_\ell)$. Fix $\epsilon > 0$. By Corollary~5 of \cite{Pinelis2016}, $r_{\sup}(d,\Gamma)$ can be approximated by $r(d,\nu)$ arbitrarily well for priors $\nu \in \Gamma$ with a finite support; that is, there exists $\nu \in \Gamma$ with finite support such that $r(d,\nu) \geq r_{\sup}(d,\Gamma) - \epsilon$. For sufficiently large $\ell$, $\Omega_\ell$ contains all support points of $\nu$ and hence $r_{\sup}(d,\tilde{\Gamma}_\ell) \geq r(d,\nu) \geq r_{\sup}(d,\Gamma) - \epsilon$. The desired result follows.
	\end{proof}
	
	\begin{lemma} \label{lemma: continuity of r_sup}
		Under Condition~\ref{condition: conditions on risk}, for any $\Gamma' \subseteq \Gamma$ and $\epsilon>0$, there exists $\delta>0$ such that $r_{\sup}(d^*,\Gamma') - r_{\sup}(d,\Gamma') \leq \epsilon$ for all $d \in \decisionspace$ such that $\varrho(d,d^*) \leq \delta$.
	\end{lemma}
	\begin{proof}[Proof of Lemma~\ref{lemma: continuity of r_sup}]
		By Corollary~5 of \cite{Pinelis2016}, there exists $\nu \in \Gamma'$ with a finite support such that $r_{\sup}(d^*,\Gamma') \leq r(d^*,\nu) + \epsilon/2$. By Condition~\ref{condition: conditions on risk} and the fact that $\nu$ has a finite support, there exists $\delta>0$ such that, for any $d \in \decisionspace$ such that $\varrho(d,d^*) \leq \delta$, $|r(d,\nu)-r(d^*,\nu)| \leq \epsilon/2$. Since $\nu \in \Gamma'$, we have that $r_{\sup}(d,\Gamma') \geq r(d,\nu)$ and thus $r_{\sup}(d^*,\Gamma') - r_{\sup}(d,\Gamma') \leq r(d^*,\nu) + \epsilon/2 - r(d,\nu) \leq \epsilon$ for any $d \in \decisionspace$ such that $\varrho(d,d^*) \leq \delta$.
	\end{proof}
	
	\begin{lemma} \label{lemma: subset of modelspace, grid approximate subset sufficient}
		Under Condition~\ref{condition: subset of modelspace}, it holds that $\lim_{i \rightarrow \infty} r_{\sup}(d,\tilde{\Gamma}_{i|\ell}) = r_{\sup}(d,\tilde{\Gamma}_\ell)$.
	\end{lemma}
	
	\begin{proof}[Proof of Lemma~\ref{lemma: subset of modelspace, grid approximate subset sufficient}]
		Let $d \in \decisionspace$, $\ell$ and $\epsilon>0$ be fixed. By
		Corollary~5 of \cite{Pinelis2016},
		$r_{\sup}(d,\tilde{\Gamma}_\ell) \leq r(d,\pi) + \epsilon/2$ for some $\pi \in \tilde{\Gamma}_\ell$ with a finite support. Under Condition~\ref{condition: subset of modelspace, grid approximate subset}, there exists a sequence $\pi_i \in \tilde{\Gamma}_{i \mid \ell}$ such that, for all sufficiently large $i$, $r(d,\pi_i) \geq r(d,\pi) - \epsilon/2$. For such $i$, $r_{\sup}(d,\tilde{\Gamma}_\ell) \leq r(d,\pi_i) + \epsilon$. Since $r_{\sup}(d,\tilde{\Gamma}_\ell) \geq r_{\sup}(d,\tilde{\Gamma}_{i \mid \ell}) \geq r(d,\pi_i)$, we have that $r(d,\pi_i) \leq r_{\sup}(d,\tilde{\Gamma}_{i \mid \ell}) \leq r_{\sup}(d,\tilde{\Gamma}_\ell) \leq r(d,\pi_i) + \epsilon$ for all sufficiently large $i$, and thus we have proved Lemma~\ref{lemma: subset of modelspace, grid approximate subset sufficient}.
	\end{proof}
	
	\begin{proof}[Proof of Theorem~\ref{theorem: approximate Gamma-minimax on a grid}]
		Let $\epsilon>0$. There exists $d' \in \decisionspace$ such that
		$$r_{\sup}(d',\Gamma) \leq \inf_{d \in \decisionspace} r_{\sup}(d,\Gamma) + \epsilon.$$
		Moreover, there exists $\pi_\ell \in \Gamma_\ell$ such that
		$$r_{\sup}(d',\Gamma_\ell) \leq r(d',\pi_\ell) + \epsilon.$$
		Using the fact that $d^*_\ell$ is $\Gamma_\ell$-minimax and the definition of $r_{\sup}$, it holds that
		\begin{align*}
			r_{\sup}(d^*_\ell,\Gamma_\ell) &\leq r_{\sup}(d',\Gamma_\ell) \leq r(d',\pi_\ell) + \epsilon \\
			&\leq r_{\sup}(d',\Gamma) + \epsilon \leq \inf_{d \in \decisionspace} r_{\sup}(d,\Gamma) + 2\epsilon.
		\end{align*}
		Since this inequality holds for any $\epsilon>0$, we have that $r_{\sup}(d^*_\ell,\Gamma_\ell) \leq \inf_{d \in \decisionspace} r_{\sup}(d,\Gamma)$. An almost identical argument shows that the sequence $\{r_{\sup}(d^*_\ell,\Gamma_\ell)\}_{\ell=1}^\infty$ is nondecreasing. Therefore, this sequence converges to some limit $\mathcal{R} \leq \inf_{d \in \decisionspace} r_{\sup}(d,\Gamma) \leq r_{\sup}(d^*,\Gamma)$.
		
		We next prove that $r_{\sup}(d^*,\Gamma) \leq \mathcal{R}$. Let $\epsilon>0$. Without loss of generality, we may assume that $\modelspace_\ell \subseteq \Omega_\ell$ for all $\ell=1,2,\ldots$ in Condition~\ref{condition: subset of modelspace}. (Otherwise, we may instead consider the sequence $\{\Omega_{\tilde{\ell}}\}_{\tilde{\ell}=1}^\infty$ where $\Omega_{\tilde{\ell}} = \bigcap_{\ell': \Omega_{\ell'} \supseteq \modelspace_\ell} \Omega_{\ell'}$. Note that Condition~\ref{condition: subset of modelspace} also holds for $\{\Omega_{\tilde{\ell}}\}_{\tilde{\ell}=1}^\infty$.) By Lemma~\ref{lemma: approximate r_sup on subset}, there exists $\ell_0$ such that $r_{\sup}(d^*,\tilde{\Gamma}_{\ell_0}) \geq r_{\sup}(d^*,\Gamma) - \epsilon/3$. By Condition~\ref{condition: subset of modelspace}, there exists $i_1$ such that $r_{\sup}(d^*,\tilde{\Gamma}_{i_1|\ell_0}) \geq r_{\sup}(d^*,\tilde{\Gamma}_{\ell_0}) - \epsilon/3$. Without loss of generality, suppose that $d^*_\ell \rightarrow d^*$ (otherwise, take a convergent subsequence to this limit point). This then implies that there exists $i_2 > i_1$ such that $\varrho(d^*_{i_2},d^*)$ is sufficiently small, such that, by Lemma~\ref{lemma: continuity of r_sup}, $r_{\sup}(d^*_{i_2},\tilde{\Gamma}_{i_1|\ell_0}) \geq r_{\sup}(d^*,\tilde{\Gamma}_{i_1|\ell_0}) - \epsilon/3$. Moreover, since $\tilde{\Gamma}_{i_1|\ell_0} \subseteq \tilde{\Gamma}_{i_1} \subseteq \tilde{\Gamma}_{i_2}$, it holds that $r_{\sup}(d^*_{i_2},\tilde{\Gamma}_{i_2}) \geq r_{\sup}(d^*_{i_2},\tilde{\Gamma}_{i_1|\ell_0})$. Therefore, $r_{\sup}(d^*_{i_2},\tilde{\Gamma}_{i_2}) \geq r_{\sup}(d^*,\Gamma) - \epsilon$. Since the sequence $\{r_{\sup}(d^*_\ell,\Gamma_\ell)\}_{\ell=1}^\infty$ is nondecreasing, it holds that $r_{\sup}(d^*_{\ell},\Gamma_{\ell}) \geq r_{\sup}(d^*,\Gamma) - \epsilon$ for all $\ell \geq i_2$. Since $\epsilon$ is arbitrary, we have that $\liminf_{\ell \rightarrow \infty} r_{\sup}(d^*_{\ell},\Gamma_{\ell}) \geq r_{\sup}(d^*,\Gamma)$, and hence $\mathcal{R} \geq r_{\sup}(d^*,\Gamma)$.
		
		Combining the results from the preceding two paragraphs, $\mathcal{R}=\inf_{d \in \decisionspace} r_{\sup}(d,\Gamma)=r_{\sup}(d^*,\Gamma)$. Consequently, $d^*$ is $\Gamma$-minimax. Moreover, as $\{r_{\sup}(d^*_{\ell},\Gamma_{\ell})\}_{\ell=1}^\infty$ increases to $\mathcal{R}$, this sequence also increases to $r_{\sup}(d^*,\Gamma)$. This concludes the proof.
	\end{proof}
	
	\begin{proof}[Proof of Corollary~\ref{corollary: convexity; approximate Gamma-minimax on a grid}]
		We first establish the strict convexity of $d \mapsto r(d,\pi)$ for any $\pi \in \Gamma$. We then establish the strict convexity of $d \mapsto r_{\sup}(d,\Gamma)$. We then establish that there is a unique minimizer of $d \mapsto r_{\sup}(d,\Gamma)$ and show that the desired result follows from Theorem~\ref{theorem: approximate Gamma-minimax on a grid}.
		
		Let $d_1,d_2 \in \decisionspace$ and $c \in (0,1)$ be arbitrary, then by the convexity of $\decisionspace$ and the strict convexity of $d \mapsto R(d,P)$ for each $P \in \modelspace$,
		\begin{align*}
			r(c d_1 + (1-c) d_2, \pi) &= \int R(c d_1 + (1-c) d_2, P) \pi(\intd P) \\
			&< \int \{c R(d_1, P) + (1-c) R(d_2, P)\} \pi(\intd P) \\
			&= c r(d_1,\pi) + (1-c) r(d_2,\pi).
		\end{align*}
		Therefore, $d \mapsto r(d,\pi)$ is strictly convex for any $\pi \in \Gamma$.
		
		Let $d_1,d_2 \in \decisionspace$ be distinct and $c \in (0,1)$ be arbitrary. Since $r_{\sup}(d,\Gamma)$ is attainable for any $d \in \decisionspace$, there exists $\tilde{\pi} \in \Gamma$ such that
		\begin{align*}
			r_{\sup}(c d_1+ (1-c) d_2, \Gamma) &= r(c d_1+ (1-c) d_2, \tilde{\pi}) \\
			&< c r(d_1,\tilde{\pi}) + (1-c) r(d_2, \tilde{\pi}) \\
			&\leq c r_{\sup}(d_1,\Gamma) + (1-c) r_{\sup}(d_2,\Gamma).
		\end{align*}
		Thus, $d \mapsto r_{\sup}(d,\Gamma)$ is strictly convex.
		
		As $d \mapsto r_{\sup}(d,\Gamma)$ is strictly convex and $\decisionspace$ is convex, this function achieves exactly one minimum on $\decisionspace$. By Theorem~\ref{theorem: approximate Gamma-minimax on a grid}, any limit point $d^*$ of $\{d^*_\ell\}_{\ell=1}^\infty$ is a minimizer of $d \mapsto r_{\sup}(d,\Gamma)$, and so the sequence has a limit point, which is also the unique $\Gamma$-minimax estimator.
	\end{proof}

	\subsection{Proof of Theorems~\ref{theorem: SGDmax convergence} \& \ref{theorem: GDmax convergence}}
	
	We prove Theorems~\ref{theorem: SGDmax convergence} and \ref{theorem: GDmax convergence} by checking that Assumptions~3.1 and 3.6 in \cite{Lin2020} are satisfied and using Theorem~E.3 and E.4 in \cite{Lin2020}, respectively. Since Assumption~3.1 is satisfied by our construction of $\hat{R}$, we focus on Assumption~3.6 for the rest of this section.
	
	Let $\modelspace_\ell = \{P_1, P_2, \ldots, P_\Lambda\} \subseteq \modelspace$. For any $\pi \in \Gamma_\ell$, let $\pi_\lambda$ denote the probability weight of $\pi$ on $P_\lambda$ ($\lambda = 1,\ldots,\Lambda$). For the rest of this section, we also use $\pi$ to denote the vector $(\pi_1,\ldots,\pi_\Lambda)$. We also use $\lesssim$ to denote less than equal to up to a universal positive constant that may depend on $\ell$. Then, straightforward calculations imply that $\nabla_\beta r(\beta,\pi) = \sum_{\lambda=1}^\Lambda \pi_\lambda \nabla_\beta R(\beta,P_\lambda)$ and $\nabla_\pi r(\beta,\pi) = (R(\beta,P_1),\ldots,R(\beta,P_\Lambda))^\top$.
	
	For each $\ell=1,2,\ldots$, for any $\beta^1,\beta^2 \in \mathcal{H}$ and $\pi^1,\pi^2 \in \Gamma_\ell$, by Conditions~\ref{condition: uniform Lipschitz on R} and \ref{condition: uniform Lipschitz on R'},
	\begin{align*}
		& \left\| \left. \nabla_\beta r(\beta,\pi) \right|_{\beta=\beta^1, \pi=\pi^1} - \left. \nabla_\beta r(\beta,\pi) \right|_{\beta=\beta^2, \pi=\pi^2} \right\| \\
		&= \left\| \sum_{\lambda=1}^\Lambda \left\{ \pi^1_\lambda \left. \nabla_\beta R(\beta,P_\lambda) \right|_{\beta=\beta_1} - \pi^2_\lambda \left. \nabla_\beta R(\beta,P_\lambda) \right|_{\beta=\beta_2} \right\} \right\| \\
		&\leq \sum_{\lambda=1}^\Lambda \pi^1_\lambda \left\| \left. \nabla_\beta R(\beta,P_\lambda) \right|_{\beta=\beta_1} - \left. \nabla_\beta R(\beta,P_\lambda) \right|_{\beta=\beta_2} \right\| + \left\| \sum_{\lambda=1}^\Lambda (\pi^1_\lambda - \pi^2_\lambda) \left. \nabla_\beta R(\beta,P_\lambda) \right|_{\beta=\beta_2} \right\| \\
		&\lesssim \| \beta^1 - \beta^2 \| + \| \pi^1 - \pi^2 \| \\
		&\lesssim \|(\beta^1,\pi^1) - (\beta^2,\pi^2)\|,
	\end{align*}
	and similarly for $\nabla_\pi r(\beta,\pi)$,
	\begin{align*}
		& \left\| \left. \nabla_\pi r(\beta,\pi) \right|_{\beta=\beta^1, \pi=\pi^1} - \left. \nabla_\pi r(\beta,\pi) \right|_{\beta=\beta^2, \pi=\pi^2} \right\| \\
		&= \left\| \left( R(\beta^1,P_1) - R(\beta^2,P_1), R(\beta^1,P_2) - R(\beta^2,P_2), \ldots, R(\beta^1,P_\Lambda) - R(\beta^2,P_\Lambda) \right)^\top \right\| \\
		&\lesssim \| \beta^1 - \beta^2 \| \leq \|(\beta^1,\pi^1) - (\beta^2,\pi^2)\|.
	\end{align*}
	This implies that for each $\ell$, the gradient of $r(\beta,\pi)$ ($\beta \in \mathcal{H}$, $\pi \in\Gamma_\ell$) is Lipschitz continuous.
	
	For each $\ell=1,2,\ldots$, for any $\beta^1,\beta^2 \in \mathcal{H}$ and $\pi \in \Gamma_\ell$, Condition~\ref{condition: uniform Lipschitz on R} implies that
	\begin{align*}
		\left| r(\beta^1,\pi) - r(\beta^2,\pi) \right| &= \left| \sum_{\lambda=1}^\Lambda \pi_\lambda \left[ R(\beta^1,P_\lambda) - R(\beta^2,P_\lambda) \right] \right| \\
		&\leq \sum_{\lambda=1}^\Lambda \pi_\lambda \left| R(\beta^1,P_\lambda) - R(\beta^2,P_\lambda) \right| \lesssim \| \beta^1 - \beta^2 \|.
	\end{align*}
	Therefore, $\beta \mapsto r(\beta,\pi)$ is Lipschitz continuous with a universal Lipschitz constant independent of $\pi \in \Gamma_\ell$.
	
	Finally, it is straightforward to check that (i) $\pi \mapsto r(\beta,\pi)$ is concave for any $\beta \in \mathcal{H}$, and (ii) $\Gamma_\ell$ is parameterized by a convex subset of a simplex in a Euclidean space, which is a convex and bounded set. These results show that Assumption~3.6 in \cite{Lin2020} is satisfied for Algorithm~\ref{algorithm: GDmax} and \ref{algorithm: SGDmax}.

	\subsection{Proof of Theorem~\ref{theorem: SGDmax variant convergence}}
	
	\begin{proof}[Proof of Theorem~\ref{theorem: SGDmax variant convergence}]
		Let $\pi_{(t),0}$ denote a maximizer of $\pi \mapsto r(\beta_{(t-1)},\pi)$. It holds that
		\begin{align*}
			0 &\leq r(\beta_{(t-1)},\pi_{(t),0}) - r(\beta_{(t-1)},\pi_{(t)}) \\
			&\leq \frac{1}{J'} \sum_{j=1}^{J'} \hat{r}(\beta_{(t-1)},\pi_{(t)},\xi_j) - \frac{1}{J'} \sum_{j=1}^{J'} \hat{r}(\beta_{(t-1)},\pi_{(t),0},\xi_j) \\
			&\quad+ r(\beta_{(t-1)},\pi_{(t),0}) - r(\beta_{(t-1)},\pi_{(t)}) \\
			&= \frac{1}{J'} \sum_{j=1}^{J'} \Bigg\{ \left[ \hat{r}(\beta_{(t-1)},\pi_{(t)},\xi_j) - \hat{r}(\beta_{(t-1)},\pi_{(t),0},\xi_j) \right] \\
			&\quad- \expect \left[ \hat{r}(\beta_{(t-1)},\pi_{(t)},\xi) - \hat{r}(\beta_{(t-1)},\pi_{(t),0},\xi) \right] \Bigg\} \\
			&\leq \sup_{\beta \in \real^D, \pi_1,\pi_2 \in \Gamma_\ell} \Bigg| \frac{1}{J'} \sum_{j=1}^{J'} \Bigg\{ \left[ \hat{r}(\beta,\pi_1,\xi_j) - \hat{r}(\beta,\pi_2,\xi_j) \right] \\
			&\quad- \expect \left[ \hat{r}(\beta,\pi_1,\xi) - \hat{r}(\beta,\pi_2,\xi) \right] \Bigg\} \Bigg|.
		\end{align*}
		Note that the right hand side does not depend on $t$. Therefore,
		\begin{align*}
			0 &\leq \sup_{t} \left\{ r(\beta_{(t-1)},\pi_{(t),0}) - \expect[r(\beta_{(t-1)},\pi_{(t)})] \right\} \\
			&\leq \expect^* \sup_{\beta \in \real^D, \pi_1,\pi_2 \in \Gamma_\ell} \Bigg| \frac{1}{J'} \sum_{j=1}^{J'} \Bigg\{ \left[ \hat{r}(\beta,\pi_1,\xi_j) - \hat{r}(\beta,\pi_2,\xi_j) \right] \\
			&\quad- \expect \left[ \hat{r}(\beta,\pi_1,\xi) - \hat{r}(\beta,\pi_2,\xi) \right] \Bigg\} \Bigg|,
		\end{align*}
		where $\expect^*$ stands for outer expectation. We may apply Corollary~9.27 in \cite{Kosorok2008} to $\mathcal{F} := \{\xi \mapsto \hat{r}(\beta,\pi,\xi): \beta \in \real^D, \pi \in \Gamma_\ell\}$ and show that $\mathcal{F} - \mathcal{F} := \{f_1-f_2: f_1,f_2 \in \mathcal{F}\} \supseteq \{\xi \mapsto \hat{r}(\beta,\pi_1,\xi) - \hat{r}(\beta,\pi_2,\xi): \beta \in \real^D, \pi_1,\pi_2 \in \Gamma_\ell\}$ is a $\Xi$-Glivenko-Cantelli class. Therefore,
		\begin{align*}
			& \Bigg \{ \sup_{\beta \in \real^D, \pi_1,\pi_2 \in \Gamma_\ell} \Bigg| \frac{1}{J'} \sum_{j=1}^{J'} \Bigg\{ \left[ \hat{r}(\beta,\pi_1,\xi_j) - \hat{r}(\beta,\pi_2,\xi_j) \right] \\
			&\quad- \expect \left[ \hat{r}(\beta,\pi_1,\xi) - \hat{r}(\beta,\pi_2,\xi) \right] \Bigg\} \Bigg| \Bigg\}^* \\
			&\leq \left\{ \sup_{f \in \mathcal{F}-\mathcal{F}} \left| \frac{1}{J'} \sum_{j=1}^{J'} \left\{ f(\xi_j) - \expect[f(\xi)] \right\} \right| \right\}^* \overset{a.s.}{\rightarrow} 0,
		\end{align*}
		as $J' \rightarrow \infty$. Here, $X^*$ stands for the minimal measurable majorant with respect to $\Xi$ of a (possibly non-measurable) mapping $X$ \citep{vandervaart1996}.
		
		By Problem~1 of Section~2.4 in \cite{vandervaart1996}, there exists a random variable $F$ such that $F \geq \sup_{f \in \mathcal{F}-\mathcal{F}} |f(\xi) - \expect[f(\xi')]|$ $\Xi$-almost surely and $\expect[F] < \infty$. Then,
		\begin{align*}
			\sup_{f \in \mathcal{F}-\mathcal{F}} \left| \frac{1}{J'} \sum_{j=1}^{J'} \left\{ f(\xi_j) - \expect[f(\xi_j)] \right\} \right| \leq F
		\end{align*}
		$\Xi$-almost surely. By dominated convergence theorem,
		\begin{align*}
			& \expect^* \sup_{\beta \in \real^D, \pi_1,\pi_2 \in \Gamma_\ell} \Bigg| \frac{1}{J'} \sum_{j=1}^{J'} \Bigg\{ \left[ \hat{r}(\beta,\pi_1,\xi_j) - \hat{r}(\beta,\pi_2,\xi_j) \right] \\
			&\quad- \expect \left[ \hat{r}(\beta,\pi_1,\xi_j) - \hat{r}(\beta,\pi_2,\xi_j) \right] \Bigg\} \Bigg| \rightarrow 0
		\end{align*}
		as $J' \rightarrow \infty$, and so does $\sup_{t} \left\{ r(\beta_{(t-1)},\pi_{(t),0}) - \expect[r(\beta_{(t-1)},\pi_{(t)})] \right\}$. Thus, for any $\zeta>0$, there exists a sufficiently large $J'$ such that $\expect[r(\beta_{(t-1)},\pi_{(t)})] \geq r(\beta_{(t-1)},\pi_{(t),0}) - \zeta$ for all $t$.
	\end{proof}

	\subsection{Proof of Theorem~\ref{theorem: fictitious play convergence}}
	
	Our proof of Theorem~\ref{theorem: fictitious play convergence} builds on that of \cite{Robinson1951}. Major modifications are needed to allow for more general definitions that can accommodate for potentially infinite spaces of pure strategies and a more careful control on a bound on $r( \overline{d}(\varpi_{(t-1)}),\pi^\dagger_{(t)}) - r( d^\dagger_{(t)},\pi_{(t-1)})$ towards the end of the proof.
	
	In this appendix, we slightly abuse the notation and use $\decisionspace$ to denote the compact set $\bar{\decisionspace}$ that contains all $d^\dagger_{(t)}$ ($t=1,2,\ldots$). We first introduce the notion of cumulative Bayes risk functions. Under Algorithm~\ref{algorithm: fictitious play}, we let $U_0: \decisionspace \rightarrow \real$ and $V_0: \Gamma_\ell \rightarrow \real$ be any two continuous functions such that
	\begin{equation}
		\label{equation: fictitious play U0 V0}
		\min_{d \in \decisionspace} U_0(d) = \max_{\pi \in \Gamma_\ell} V_0(\pi)
	\end{equation}
	and recursively define
	\begin{equation}
		\label{equation: fictitious play def U V}
		U_{t+1}(d) := U_{t}(d) + r(d,\pi^\dagger_{(t)}), \quad V_{t+1}(\pi) := V_{t}(\pi) + r(d^\dagger_{(t)},\pi)
	\end{equation}
	for $d \in \decisionspace$ and $\pi \in \Gamma_\ell$. Here, we let $\pi^\dagger_{(t)} \in \argmax_{\pi \in \Gamma_\ell} V_{t-1}(\pi)$ and $d^\dagger_{(t)} \in \argmin_{d \in \decisionspace} U_{t-1}(d)$. Note that the choices of $\pi_{(t)}^\dagger$ and $d_{(t)}$ in Algorithm~\ref{algorithm: fictitious play} corresponds to setting $U_0 \equiv 0$ and $V_0 \equiv 0$, in which case $U_{t}(d) = t \cdot r(d,\pi_{(t)})$ and $V_{t}(\pi) = t \cdot r(\overline{d}(\varpi_{(t)}),\pi)$. In general,
	\begin{equation}
		\label{equation: fictitious play Ut Vt as t r}
		U_{t}(d) = U_0(d) + t \cdot r(d,\pi_{(t)}), \quad V_{t}(\pi) = V_0(\pi) + t \cdot r(\overline{d}(\varpi_{(t)}),\pi)
	\end{equation}
	for some $\pi_{(t)} \in \Gamma$ and $\overline{d}(\varpi_{(t)}) \in \overline{\decisionspace}$. Later in this section, we will also make use of $U_t$ and $V_t$ with other initializations $U_0$ and $V_0$.
	
	To make notations concise, we define $\min_{d \in \decisionspace'} U_{t} := \min_{d \in \decisionspace'} U_{t}(d)$ for any $\decisionspace' \subseteq \decisionspace$, and define $\max_{\decisionspace'} U_{t}$, $\min_{\Pi'}{V}_{t}$ and $\max_{\Pi'} V_{t}$ ($\Pi' \subseteq \Gamma_\ell$) similarly. We also drop the subscript denoting the set when the set is the whole space we consider, that is, $\decisionspace$ or $\Gamma_\ell$. Note that for any $t_1,t_2=1,2,\ldots$, under the setting of Algorithm~\ref{algorithm: fictitious play} and \eqref{equation: fictitious play minimax theorem}, it holds that
	\begin{align*}
		&\quad \min U_{t_1}/t_1 = \min_{\overline{d} \in \overline{\decisionspace}} r(\overline{d},\pi_{(t_1)}) \\
		&\leq \max_{\pi \in \Gamma_{\ell}} \min_{\overline{d} \in \overline{\decisionspace}} r(\overline{d},\pi) = r(\overline{d}(\varpi_\ell^*),\pi^*_\ell) = \min_{\overline{d} \in \overline{\decisionspace}} \max_{\pi \in \Gamma_{\ell}} r(\overline{d},\pi) \\
		&\leq \max_{\pi \in \Gamma_{\ell}} r(\overline{d}(\varpi_{(t_2)}),\pi) = \max V_{t_2}/t_2
	\end{align*}
	Therefore, to prove the first result in Theorem~\ref{theorem: fictitious play convergence}, it suffices to show that $\limsup_{t \rightarrow \infty} (\max V_{t} - \min U_{t})/t \leq 0$.
	
	We next introduce additional definitions related to iterations. We say that $\pi \in \Gamma_\ell$ is eligible in the interval $[t_1,t_2]$ if there exists $t \in [t_1,t_2]$ such that $V_{t}(\pi) = \max V_{t}$; we say that $d \in \decisionspace$ is eligible in the interval $[t_1,t_2]$ if there exists $t \in [t_1,t_2]$ such that $U_{t}(d) = \min U_{t}$. We also define eligibility for sets. We say that $\Pi' \subseteq \Gamma_\ell$ is eligible in the interval $[t_1,t_2]$ if there exists $\pi \in \Pi'$ that is eligible in that interval; we say that $\decisionspace' \subseteq \decisionspace$ is eligible in the interval $[t_1,t_2]$ if there exists $d \in \decisionspace'$ that is eligible in the interval $[t_1,t_2]$. In addition, for any $\decisionspace' \subseteq \decisionspace$, we define maximum variation $\MV_{t}(\decisionspace') := \sup_{d \in \decisionspace'} U_{t}(d) - \inf_{d \in \decisionspace'} U_{t}(d)$ and $\MV_{t}(\Pi')$ similarly for any $\Pi' \subset \Gamma_\ell$. By Condition~\ref{condition: conditions on risk}, there exists $B \in (0,\infty)$ such that $R \in [-B,B]$. Note that by Condition~\ref{condition: limit point} and Lemma~\ref{lemma: continuity of r_sup}, given an arbitrary desired approximation accuracy $\epsilon>0$, $\decisionspace$ can be covered by finitely many compact subsets with the maximum variation of each subset bounded by $\epsilon t$ for all $t$; by Condition~\ref{condition: conditions on risk}, since $\Gamma_\ell$ is parameterized by a compact subset of a simplex in a Euclidean space, $\Gamma_\ell$ can also be covered by finitely many compact subsets with the maximum variation of each subset bounded by $\epsilon t$ for all $t$. These covers can be viewed as discrete finite approximations to $\decisionspace$ and $\Gamma_\ell$, respectively.
	
	All of the above definitions are associated with the space of estimators $\decisionspace$ and the set of priors $\Gamma_\ell$. We call $\{(U_t,V_t)\}_t$ a pair of cumulative Bayes risk functions constructed from the pair $(\decisionspace,\Gamma_\ell)$ of the space of estimators and the set of priors, and will consider pairs of cumulative Bayes risk functions constructed from other pairs $(\decisionspace',\Pi')$ of the space of estimators and the set of priors in the subsequent proof. We can define the above quantities similarly for such cases.
	
	The following lemma gives an upper bound on the maximum variation of $U_{s+t}$ and $V_{s+t}$ over the corresponding entire space from which they are constructed after $t$ iterations from $s$ when essentially all parts of these spaces are eligible in $[s,s+t]$.
	
	\begin{lemma} \label{lemma: fictitious play lemma 1}
		Suppose that $\{(U_t,V_t)\}_t$ is a pair of cumulative Bayes risk functions constructed from $(\decisionspace',\Pi')$. Suppose that $\decisionspace' = \bigcup_{i=1}^I \decisionspace_i$ and $\Pi' = \bigcup_{j=1}^J \Pi_j$ where
		$$\sup_{i,t} \MV_{t}(\decisionspace_i)/t \leq A, \quad \sup_{j,t} \MV_{t}(\Pi_j)/t \leq A$$
		for $A<\infty$. If all $\decisionspace_i$ and $\Pi_j$ are eligible in $[s,s+t]$, then $\max_{\decisionspace'} U_{s+t} - \min_{\decisionspace'} U_{s+t} \leq (2B+A) t$ and $\max_{\Pi'} V_{s+t} - \min_{\Pi'} V_{s+t} \leq (2B+A) t$.
	\end{lemma}
	\begin{proof}[Proof of Lemma~\ref{lemma: fictitious play lemma 1}]
		Without loss of generality, assume that $\tilde{d} \in (\argmax_{d \in \decisionspace'} U_{s+t}) \bigcap \decisionspace_1$. Since $\decisionspace_1$ is eligible in $[s,t]$, there exists $\tilde{t} \in [s,s+t]$ such that $(\argmin_{d \in \decisionspace'} U_{\tilde{t}}) \bigcap \decisionspace_1 \neq \emptyset$. By the recursive definition of the sequence $\{U_{t}\}_t$ in \eqref{equation: fictitious play def U V}, the bound on the risk, and the assumption that $\sup_{i,t} \MV_t(\decisionspace_i)/t \leq A$, we have that $\max_{\decisionspace'} U_{s+t} = U_{s+t}(\tilde{d}) \leq U_{\tilde{t}}(\tilde{d}) + B(s+t-\tilde{t}) \leq \min_{\decisionspace'} U_{\tilde{t}} + At + B(s+t-\tilde{t}) \leq \min_{\decisionspace'} U_{\tilde{t}} + (A+B) t$. Letting $\tilde{d}' \in \argmin_{d \in \decisionspace'} U_{s+t}$, by the bound on the risk, we can derive that $\min_{\decisionspace'} U_{s+t} = U_{s+t}(\tilde{d}') \geq U_{\tilde{t}}(\tilde{d}') - B(s+t-\tilde{t}) \geq \min_{\decisionspace'} U_{\tilde{t}} - Bt$. Combine these two inequalities and we have that $\max_{\decisionspace'} U_{s+t} - \min_{\decisionspace'} U_{s+t} \leq (2B+A) t$. An identical argument applied to the sequence $\{V_t\}_t$ shows that $\max_{\Pi'} V_{s+t} - \min_{\Pi'} V_{s+t} \leq (2B+A) t$.
	\end{proof}
	
	The next lemma builds on the previous lemma and provides an upper bound on $\max V_{s+t} - \min U_{s+t}$ under the same conditions.
	
	\begin{lemma} \label{lemma: fictitious play lemma 2}
		Under the same setup and conditions as in Lemma~\ref{lemma: fictitious play lemma 1}, $\max_{\Pi'} V_{s+t} - \min_{\decisionspace'} U_{s+t} \leq (4B+2A)t$.
	\end{lemma}
	\begin{proof}[Proof of Lemma~\ref{lemma: fictitious play lemma 2}]
		Summing the two inequalities in Lemma~\ref{lemma: fictitious play lemma 1} and rearranging the terms, we have that $\max_{\Pi'} V_{s+t} - \min_{\decisionspace'} U_{s+t} \leq (4B+2A)t + \min_{\Pi'} V_{s+t} - \max_{\decisionspace'} U_{s+t}$. It therefore suffices to show that $\min_{\Pi'} V_{s+t} \leq \max_{\decisionspace'} U_{s+t}$.
		
		Let $\tau:=s+t$. There exists $\pi' \in \Pi'$ and a stochastic strategy $\overline{d}' \in \decisionspace'$ such that $U_\tau(d) = U_0(d) + \tau \cdot r(d,\pi')$ and $V_\tau(\pi) = V_0(\pi) + \tau \cdot r(\overline{d}',\pi)$ for all $d \in \decisionspace'$ and all $\pi \in \Pi'$. Therefore, for this choice of $\pi'$ and $\overline{d}'$, using \eqref{equation: fictitious play U0 V0}, $\min_{\Pi'} V_{\tau} \leq V_{\tau}(\pi') = V_0(\pi') + \tau \cdot r(\overline{d}',\pi') \leq \max_{\Pi'} V_0 + \tau \cdot r(\overline{d}',\pi') = \min_{\decisionspace'} U_0 + \tau \cdot r(\overline{d}',\pi') \leq U_0(\overline{d}') + \tau \cdot r(\overline{d}',\pi') = U_{\tau}(\overline{d}') \leq \max_{\decisionspace'} U_{\tau}$.
	\end{proof}
	
	\begin{proof}[Proof of Theorem~\ref{theorem: fictitious play convergence}]
		It suffices to show that $\limsup_{t \rightarrow \infty} (\max V_{t} - \min U_{t})/t \leq 0$ by letting $U_0 \equiv 0$ and $V_0 \equiv 0$, which corresponds to Algorithm~\ref{algorithm: fictitious play}. Let $\epsilon>0$. Note that $r$ is Lipschitz continuous by Lemma~\ref{lemma: continuity of r_sup} and the fact that $r(d,\pi)$ is an average of bounded risks with weights $\pi$. Furthermore, $\decisionspace$ and $\Gamma_\ell$ are both compact. In addition, $U_0$ and $V_0$ are both continuous. Therefore, there exist covers $\decisionspace = \bigcup_{i=1}^I \decisionspace_i$ and $\Gamma_\ell = \bigcup_{j=1}^J \Pi_j$ such that (i) $\decisionspace_i$ and $\Pi_j$ are all compact, and (ii) $\sup_{i,t} \MV_t(\decisionspace_i)/t \leq \epsilon$, $\sup_{j,t} \MV_t(\Pi_j)/t \leq \epsilon$. (Note that $I$ and $J$ may depend on $\epsilon$.) For index sets $\mathcal{I} \subseteq \{1,2,\ldots,I\}$ and $\mathcal{J} \subseteq \{1,2,\ldots,J\}$, define $\decisionspace_{\mathcal{I}} := \bigcup_{i \in \mathcal{I}} \decisionspace_i$ and $\Pi_{\mathcal{J}} := \bigcup_{j \in \mathcal{J}} \Pi_j$. We show that $\max V_{t} - \min U_{t} \leq C \epsilon t$ for an absolute constant $C$ and all sufficiently large $t$ via induction on the sizes of $\mathcal{I}$ and $\mathcal{J}$.
		
		Let $\{(U_t,V_t)\}_t$ be a pair of cumulative Bayes risk functions constructed from $(\decisionspace_\mathcal{I},\Pi_\mathcal{J})$ where $|\mathcal{I}|=|\mathcal{J}|=1$. By \eqref{equation: fictitious play Ut Vt as t r} and the fact that $\MV_t(\decisionspace_\mathcal{I}) \leq \epsilon t$ and $\MV_t(\Pi_\mathcal{J}) \leq \epsilon t$, we have that
		\begin{align*}
			\min_{\decisionspace_\mathcal{I}} U_t &= \min_{d \in \decisionspace_\mathcal{I}} [U_0(d) + t \cdot r(d,\pi_{(t)})] \geq \expect_{d \sim \varpi_{(t)}}[U_0(d)] + t \cdot r(\overline{d}(\varpi_{(t)}),\pi_{(t)}) - \epsilon t \\
			&\geq \min_{d \in \decisionspace_\mathcal{I}} U_0(d) + t \cdot r(\overline{d}(\varpi_{(t)}),\pi_{(t)}) - \epsilon t \\
			&= \max_{\pi \in \Pi_{\mathcal{J}}} V_0(\pi) + t \cdot r(\overline{d}(\varpi_{(t)}),\pi_{(t)}) - \epsilon t \\
			&\geq V_0(\pi_{(t)}) + t \cdot r(\overline{d}(\varpi_{(t)}),\pi_{(t)}) - \epsilon t \\
			&\geq \max_{\pi \in \Pi_{\mathcal{J}}} [V_0(\pi) + t \cdot r(\overline{d}(\varpi_{(t)}),\pi)] - 2 \epsilon t = \max_{\Pi_{\mathcal{J}}} V_t - 2 \epsilon t.
		\end{align*}
		Therefore, $\max_{\Pi_{\mathcal{J}}} V_t - \min_{\decisionspace_\mathcal{I}} U_t \leq 2 \epsilon t$.
		
		Let $\epsilon'>0$ be arbitrary. Suppose that there exists $t_0$ such that, for any $\mathcal{I}' \subseteq \mathcal{I}$ and $\mathcal{J}' \subseteq \mathcal{J}$ such that $\mathcal{I}' \neq \mathcal{I}$ or $\mathcal{J}' \neq \mathcal{J}$, for any pair of cumulative Bayes risk functions $\{(U_t,V_t)\}_t$ constructed from $(\decisionspace_{\mathcal{I}'},\Pi_{\mathcal{J}'})$, it holds that $\max_{\Pi_{\mathcal{J}'}} V_t - \min_{\decisionspace_{\mathcal{I}'}} U_t \leq \epsilon' t$ for all $t \geq t_0$. We next obtain a slightly greater bound on $\max_{\Pi_{\mathcal{J}}} V_t - \min_{\decisionspace_{\mathcal{I}}} U_t$ for all sufficiently large $t$.
		
		We first prove that if, for a given pair of cumulative Bayes risk functions $\{(U_t,V_t)\}_t$ constructed from $(\decisionspace_{\mathcal{I}},\Pi_{\mathcal{J}})$, there exists $i' \in \mathcal{I}$ or $j' \in \mathcal{J}$ such that $\decisionspace_{i'}$ or $\Pi_{j'}$ is not eligible in an interval $[s,s+t_0]$, then
		\begin{equation}
			\label{equation: fictitious play proof induction lemma}
			\max_{\Pi_{\mathcal{J}}} V_{s+t_0} - \min_{\decisionspace_\mathcal{I}} U_{s+t_0} \leq \max_{\Pi_{\mathcal{J}}} V_{s} - \min_{\decisionspace_\mathcal{I}} U_{s} + \epsilon' t_0.
		\end{equation}
		
		Suppose that $\decisionspace_{i'}$ is not eligible in $[s,s+t_0]$, then define $U_t' := U_{s+t}$ and $V_t' := V_{s+t} - \max_{\Pi_{\mathcal{J}}} V_s + \min_{\decisionspace_\mathcal{I}} U_s$ for all $t \geq 0$. It is straightforward to check that $\{(U_t',V_t')\}_{t=0}^{t_0}$ satisfies the recursive definition of a pair of cumulative Bayes risk functions constructed from $(\decisionspace_{\mathcal{I} \setminus \{i'\}},\Pi_{\mathcal{J}})$.  By the induction hypothesis, $\max_{\Pi_{\mathcal{J}}} V_{t_0}' - \min_{\decisionspace_{\mathcal{I} \setminus \{i'\}}} U_{t_0}' \leq \epsilon' t_0$. Therefore, $\max_{\Pi_{\mathcal{J}}} V_{s+t_0} - \min_{\decisionspace_{\mathcal{I}}} U_{s+t_0} = \max_{\Pi_{\mathcal{J}}} V_{t_0}' - \min_{\decisionspace_{\mathcal{I} \setminus \{i'\}}} U_{t_0}' + \max_{\Pi_{\mathcal{J}}} V_s - \min_{\decisionspace_{\mathcal{I}}} U_s \leq \max_{\Pi_{\mathcal{J}}} V_s - \min_{\decisionspace_{\mathcal{I}}} U_s + \epsilon' t_0$. Similar argument can be applied if $\Pi_{j'}$ is not eligible in $[s,s+t_0]$.
		
		Now we obtain a bound on $\max_{\Pi_{\mathcal{J}}} V_t - \min_{\decisionspace_{\mathcal{I}}} U_t$. Let $t > t_0$, $\mathscr{Q} := \lfloor t/t_0 \rfloor \geq 1$ and $\mathscr{R} := t/t_0 - \mathscr{Q} \in [0,1)$. There are two cases.
		
		\noindent \textbf{Case 1}: There exists $s_0 \leq \mathscr{Q}$ such that $\decisionspace_i$ and $\Pi_j$ are eligible in $[(s_0 - 1 + \mathscr{R}) t_0, (s_0 + \mathscr{R}) t_0]$ for all $i \in \mathcal{I}$ and $j \in \mathcal{J}$. Take $s_0$ to be the largest such integer. Then, repeatedly apply \eqref{equation: fictitious play proof induction lemma} to intervals $[(s_0 + \mathscr{R}) t_0, (s_0 +1+ \mathscr{R}) t_0], [(s_0 + 1 + \mathscr{R}) t_0, (s_0 +2+ \mathscr{R}) t_0], \ldots, [(\mathscr{Q}-1+\mathscr{R})t_0,(\mathscr{Q}+\mathscr{R})t_0]=[t-t_0,t]$ and we derive that
		$$\max_{\Pi_{\mathcal{J}}} V_{t} - \min_{\decisionspace_{\mathcal{I}}} U_{t} \leq \max_{\Pi_{\mathcal{J}}} V_{(s_0+\mathscr{R})t_0} - \min_{\decisionspace_{\mathcal{I}}} U_{(s_0+\mathscr{R})t_0} + \epsilon' (\mathscr{Q}-s_0) t_0.$$
		By Lemma~\ref{lemma: fictitious play lemma 2}, $\max_{\Pi_{\mathcal{J}}} V_{(s_0+\mathscr{R})t_0} - \min_{\decisionspace_{\mathcal{I}}} U_{(s_0+\mathscr{R})t_0} \leq (4B+\epsilon)t_0$. Therefore,
		$$\max_{\Pi_{\mathcal{J}}} V_{t} - \min_{\decisionspace_{\mathcal{I}}} U_{t} \leq (4B+\epsilon)t_0 + \epsilon' (\mathscr{Q}-s_0) t_0 \leq (4B+\epsilon) t_0 + \epsilon' t.$$
		
		\noindent \textbf{Case 2}: There is no integer $s_0$ satisfying the condition in Case 1. Then, repeatedly apply \eqref{equation: fictitious play proof induction lemma} to intervals $[\mathscr{R} t_0, (1+\mathscr{R}) t_0], [(1+\mathscr{R}) t_0, (2+\mathscr{R}) t_0],\ldots, [(\mathscr{Q}-1+\mathscr{R}) t_0, (\mathscr{Q}+\mathscr{R}) t_0]=[t-t_0,t]$, we derive that
		$$\max_{\Pi_{\mathcal{J}}} V_{t} - \min_{\decisionspace_{\mathcal{I}}} U_{t} \leq \max_{\Pi_{\mathcal{J}}} V_{\mathscr{R}t_0} - \min_{\decisionspace_{\mathcal{I}}} U_{\mathscr{R}t_0} + \epsilon' \mathscr{Q} t_0.$$
		By the bound on the risk, $\max_{\Pi_{\mathcal{J}}} V_{\mathscr{R}t_0} \leq B \mathscr{R} t_0$ and $\min_{\decisionspace_\mathcal{I}} U_{\mathscr{R}t_0} \geq -B \mathscr{R} t_0$. Hence,
		$$\max_{\Pi_{\mathcal{J}}} V_{t} - \min_{\decisionspace_{\mathcal{I}}} U_{t} \leq 2B \mathscr{R} t_0 + \epsilon' \mathscr{Q} t_0 \leq (4B+\epsilon) t_0 + \epsilon' t.$$
		
		Thus, in both cases, it holds that $\max_{\Pi_{\mathcal{J}}} V_{t} - \min_{\decisionspace_{\mathcal{I}}} U_{t} \leq (4B+\epsilon) t_0 + \epsilon' t$ for $t > t_0$. Let $C>0$ be any constant (which may depend on $\epsilon$, the approximation error of the covers, that is, the bound on $\MV_t/t$). The following holds for any sufficiently large $t$,
		\begin{equation}
			\label{equation: fictitious play proof induction}
			\max_{\Pi_{\mathcal{J}}} V_{t} - \min_{\decisionspace_{\mathcal{I}}} U_{t} \leq (4B+\epsilon) t_0 + \epsilon' t \leq (1+C) \epsilon' t.
		\end{equation}
		In other words, we show that after increasing the size of either index set by 1, for all sufficiently large $t$, we obtain a bound on $\max_{\Pi_{\mathcal{J}}} V_{t} - \min_{\decisionspace_{\mathcal{I}}} U_{t}$ that grows by a multiplicative factor of $(1+C)$ relative to the original bound.
		
		It takes finitely many, say $N$, steps to induct from the initial case where the sizes of both index sets are one to the case of interest with index sets $\{1,\ldots,I\}$ and $\{1,\ldots,J\}$. (Note that $N$ may also depend on $\epsilon$ through its dependence on $I$ and $J$.) Take $C=1/N$ in \eqref{equation: fictitious play proof induction} and we derive that, for all sufficiently large $t$,
		$$\max V_{t} - \min U_{t} = \max_{\Pi_{\{1,\ldots,J\}}} V_{t} - \min_{\decisionspace_{\{1,\ldots,I\}}} U_{t} \leq (1+1/N)^N \cdot 2 \epsilon t \leq 2 e \epsilon t$$
		where $e$ is the base of natural logarithm. Since $\epsilon$ is arbitrary, we show that $\limsup_{t \rightarrow \infty} (\max V_{t} - \min U_{t})/t \leq 0$.
	\end{proof}

	\subsection{Derivation of $\Gamma$-minimax estimator of the mean in Section~\ref{section: simulation mean}} \label{section: theory of mean estimation}
	
	In this section, we show that, for the problem of estimating the mean in Section~\ref{section: simulation mean}, one $\Gamma$-minimax estimator lies in $\decisionspace_\mathrm{linear}$. This is formally presented below.
	
	\begin{proposition} \label{proposition: Gamma-minimax estimator of mean}
		Let $\modelspace$ consist of all probability distributions defined on the Borel $\sigma$-algebra on $[0,1]$. Let $X_1,\ldots,X_n \overset{\mathrm{iid}}{\sim} P_0 \in \modelspace$ and $\observation=(X_1,X_2,\ldots,X_n)$ be the observed data. Let $\Psi: P \mapsto \int x P(\intd x)$ denote the mean parameter and $\Gamma=\{\pi \in \Pi: \int \Psi(P) \pi(\intd P) = \mu\}$ be the set of priors that represent prior information. Let $\decisionspace$ denote the space of estimators that are square-integrable with respect to all $P \in \modelspace$. Consider the risk in Example~\ref{example: estimation}, $R: (d,P) \mapsto \expect_P[(d(\observation) - \Psi(P))^2]$. Define $\bar{X} = \sum_{i=1}^n X_i/n$ and $d_0: \observation \mapsto (\mu + \sqrt{n} \bar{X})/(1+\sqrt{n})$. Then $d_0 \in \decisionspace_\mathrm{linear}$ is $\Gamma$-minimax over $\decisionspace$.
	\end{proposition}
	
	We first present a theorem on a criterion of $\Gamma$-minimaxity.
	\begin{theorem} \label{theorem: Gamma-minimaxity criterion}
		Suppose that $d_0 \in \decisionspace$ is a Bayes estimator for $\pi_0 \in \Gamma$ and $r(d_0,\pi_0)=r_{\sup}(d_0,\Gamma)$. Then $d_0$ is a $\Gamma$-minimax estimator in $\decisionspace$.
	\end{theorem}
	\begin{proof}[Proof of Theorem~\ref{theorem: Gamma-minimaxity criterion}]
		Clearly $r_{\sup}(d_0,\Gamma) \geq \inf_{d \in \decisionspace} r_{\sup}(d,\Gamma)$. Fix $d' \in \decisionspace$. Then, $r_{\sup}(d',\Gamma) \geq r(d',\pi_0) \geq r(d_0,\pi_0) = r_{\sup}(d_0,\Gamma)$. Since $d'$ is arbitrary, this shows that $\inf_{d \in \decisionspace} r_{\sup}(d,\Gamma) \geq r_{\sup}(d_0,\Gamma)$. Thus, $r_{\sup}(d_0,\Gamma) = \inf_{d \in \decisionspace} r_{\sup}(d,\Gamma)$ and $d_0$ is $\Gamma$-minimax.
	\end{proof}
	
	We now present a lemma that is used to prove Proposition~\ref{proposition: Gamma-minimax estimator of mean}.
	
	\begin{lemma} \label{lemma: max variance distribution}
		Let $a<b$ and suppose that $\modelspace$ denotes the model space that consists of all probability distributions defined on the Borel $\sigma$-algebra on $[a,b] \subseteq \real$ with mean $\mu \in [a,b]$. Let $X$ denote a generic random variable generated from some $P \in \modelspace$. Then $\max_{P \in \modelspace} \Var_P(X) = \Var_{P^*}(X) = (b-\mu)(\mu-a)$, where $P^*$ is defined by $P^*(X = a) = (b-\mu)/(b-a)$ and $P^*(X = b) = (\mu-a)/(b-a)$.
	\end{lemma}
	\begin{proof}[Proof of Lemma~\ref{lemma: max variance distribution}]
		Without loss of generality, we may assume that $a=-1$ and $b=1$.
		Note that for any $P \in \modelspace$, it holds that $\Var_P(X) = \expect_P[X^2] - \expect_P[X]^2 = \expect_P[X^2] - \mu^2 \leq 1 - \mu^2$, where the equality is attained if $P(X \in \{-1,1\})=1$. Therefore, the maximum variance is achieved at the distribution with the specified mean $\mu$ and support being $\{a,b\}$, that is, at the distribution $P^*$ defined in the lemma statement. Straightforward calculations show that $\Var_{P^*}(X) = (b-\mu)(\mu-a)$.
	\end{proof}
	
	\begin{proof}[Proof of Proposition~\ref{proposition: Gamma-minimax estimator of mean}]
		Let $\modelspace':=\{\text{Bernoulli}(\theta): \theta \in (0,1)\} \subseteq \modelspace$ and let $\pi_0$ be a prior distribution over $\modelspace'$ such that the prior distribution on the success probability $\theta$ is $\mathrm{Beta}(\mu \sqrt{n}, (1-\mu) \sqrt{n})$. By Theorem~1.1 in Chapter~4 of \cite{Lehmann1998}, a Bayes estimator for $\pi_0$ minimizes the risk under the posterior distribution, whose minimizer over $\decisionspace$ is the posterior mean $d_0$ for our choice of risk. That is, $d_0$ is a Bayes estimator in $\decisionspace$ for $\pi_0$.
		
		We next show that $r(d_0,\pi_0)=\sup_{\pi \in \Gamma} r(d_0,\pi)$. Let $\pi \in \Gamma$ be arbitrary. Since $\expect_P[\bar{X}]=\Psi(P)$ and $\Var_P(\bar{X})=\Var_P(X_1)/n$, we can derive that
		\begin{align*}
			r(d_0,\pi) &= \int \expect_P\left[\left\{ \frac{\mu + \sqrt{n} \bar{X}}{1+\sqrt{n}} - \Psi(P) \right\}^2 \right] \pi(\intd P) \\
			&= \int \expect_P\left[\left\{ \frac{\sqrt{n}}{1+\sqrt{n}} \left(\bar{X} - \Psi(P)\right) + \frac{\mu-\Psi(P)}{1+\sqrt{n}} \right\}^2 \right] \pi(\intd P) \\
			&= \int \left\{ \frac{1}{(1+\sqrt{n})^2} \Var_P(X_1) + \frac{(\mu-\Psi(P))^2}{(1+\sqrt{n})^2} \right\} \pi(\intd P)
			\intertext{Apply Lemma~\ref{lemma: max variance distribution} to $\Var_P(X_1)$ and the display continues as}
			&\leq \int \left\{ \frac{1}{(1+\sqrt{n})^2} \Psi(P) (1-\Psi(P)) + \frac{(\mu-\Psi(P))^2}{(1+\sqrt{n})^2} \right\} \pi(\intd P) \\
			&= \int \frac{1}{(1+\sqrt{n})^2} \left\{ \mu^2 + (1-2\mu) \Psi(P) \right\} \pi(\intd P) = \frac{\mu (1-\mu)}{(1+\sqrt{n})^2}.
		\end{align*}
		This upper bound can be attained by any $\pi$ with support contained in $\modelspace'$, for example, $\pi_0$. Therefore, $r_{\sup}(d_0,\Gamma) = r(d_0,\pi_0)$. By Theorem~\ref{theorem: Gamma-minimaxity criterion}, $d_0$ is $\Gamma$-minimax over $\decisionspace$.
	\end{proof}
	
	\begin{table}[bt!]
		\centering
		\caption{Summary of frequently used symbols}
		\label{table: symbols}
		\begin{tabular}{l|l}
			\hline
			Symbol &  \\
			\hline
			$P_0$ & True data-generating mechanism \\
			$\modelspace$ & Space of data-generating mechanisms containing $P_0$ \\
			$\observation^*$ and $\observation=\coarsen(\observation^*)$ & Full generated data and coarsened data \\
			$\decisionspace$ & Space of candidate estimators or decision functions (e.g., neural networks) \\
			$R$ & Risk function \\
			$r$ & Bayes risk function $r: (d,\pi) \mapsto \int R(d,P) \pi(\intd P)$ \\
			$\Gamma (\subseteq \Pi)$ & Set of prior distributions consistent with prior knowledge \\
			$\Psi$ & Functional defining the estimand $\Psi(P_0)$ in Examples~\ref{example: estimation}--\ref{example: estimate entropy} \\
			$\modelspace_\ell$ & An increasing sequence of finite subsets of $\modelspace$ \\
			$\Gamma_\ell$ & Set of priors in $\Gamma$ with support in $\modelspace_\ell$ \\
			$r_{\sup}$ & Worst-case Bayes risk function $r_{\sup}: (d,\Gamma') \mapsto \sup_{\pi \in \Gamma'} r(d,\pi)$ \\
			$d_\ell^*$ & $\Gamma_\ell$-minimax estimator in $\decisionspace$ \\
			$d^*$ & A limit point of sequence $\{d_\ell^*\}_{\ell=1}^\infty$, which is $\Gamma$-minimax in $\decisionspace$ by Theorem~\ref{theorem: approximate Gamma-minimax on a grid} \\
			$\beta (\in \real^D)$ & Coefficient of a finite-dimensional estimator (e.g., neural network) \\
			$\xi \sim \Xi$ & Exogenous randomness \\
			$\hat{R}(\beta,P,\xi)$ & Unbiased approximation of $R(\beta,P)$ \\
			$\overline{d}(\varpi)$ & Stochastic estimator following distribution $\varpi$ over $\decisionspace$ \\
			$\overline{\decisionspace}$ & Space of stochastic estimators $\overline{d}(\varpi)$ \\
			$\overline{d}^*=\overline{d}(\varpi^*)$ & $\Gamma$-minimax estimator in $\overline{\decisionspace}$ \\
			\hline
		\end{tabular}
	\end{table}
\end{appendix}

\bibliographystyle{chicago}
\bibliography{references}

\end{document}